%% file: main.tex
\documentclass[preprint,12pt]{elsarticle}
\pdfoutput=1
\usepackage{graphicx}
\usepackage{caption}
\usepackage[labelformat=simple]{subcaption}

\usepackage{amsmath}
\usepackage{amssymb}
\usepackage{amsthm}
\usepackage{setspace}
 \usepackage[normalem]{ulem} 
\usepackage{cancel}
\usepackage[linesnumbered, ruled, vlined, noend]{algorithm2e}
\usepackage[]{color} 
\usepackage[]{url}
\usepackage[colorlinks]{hyperref}

\usepackage{tikz}
\usepackage{standalone}
\usetikzlibrary{arrows}
\usetikzlibrary{decorations} 
\usetikzlibrary{decorations.markings}
\usetikzlibrary{decorations,fit,shapes} 
\usetikzlibrary{arrows.meta,positioning}
\usetikzlibrary{shapes.misc}

\newcommand{\HR}{\mbox{{\sf HR}}}
\newcommand{\HRTLQ}{\mbox{{\sf HR2LQ}}}

\newcommand{\prefa}{\mbox{{\sf Pref($a$)}}}
\newcommand{\preflqa}{\mbox{{\sf PrefLQ($a$)}}}
\newcommand{\prefb}{\mbox{{\sf Pref($b$)}}}
\newcommand{\prefu}{\mbox{{\sf Pref($u$)}}}
\newcommand{\prefv}{\mbox{{\sf Pref($v$)}}}

\newcommand{\HRLQ}{\mbox{{\sf HRLQ}}}

\newcommand{\Df}[1]{\mbox{{\sf $def(#1)$}}}
\newcommand{\Dfa}[1]{\mbox{{\sf $def_{\mathcal{A}}(#1)$}}}
\newcommand{\Dfb}[1]{\mbox{{\sf $def_{\mathcal{B}}(#1)$}}}

\renewcommand{\AA}{\mathcal{A}}
\newcommand{\BB}{\mathcal{B}}

\newcommand{\AALQ}{\mathcal{A}_{lq}}
\newcommand{\BBLQ}{\mathcal{B}_{lq}}

\newcommand{\clr}{\textcolor{blue}}

\renewcommand{\S}{{s}}
\newcommand{\T}{{t}}

\newcommand{\cor}{{\bf corr}}

\newcommand{\tL}{\mathcal{L'}}

\newcommand{\dL}{\mathcal{D}'}
\newcommand{\dLA}{\mathcal{D}_A'}
\newcommand{\dLB}{\mathcal{D}_B'}
\newcommand{\tLA}{\mathcal{L}'_A}
\newcommand{\tLB}{\mathcal{L}'_B}

\newtheorem{cl}{Claim}
\newtheorem{definition}{Definition}
\newtheorem{theorem}{Theorem}
\newtheorem{lemma}[theorem]{Lemma}

\definecolor{lightgray}{rgb}{0.83, 0.83, 0.82}

\tikzset{cross/.style={cross out, draw=black, minimum size=2*(#1-\pgflinewidth), inner sep=0pt, outer sep=0pt},
	cross/.default={1pt}}
\usetikzlibrary{arrows, decorations.pathmorphing}

\tikzstyle{vertex}=[auto=left,circle,draw=black!80,fill=none,minimum size=15pt,inner sep=0pt]

\tikzset{
    photon/.style={decorate, decoration={snake}, draw=red}}

\newcommand{\etal}{\textit{et al}.}

\newenvironment{appendix-lemma}[1]{\vspace{0.1in}\noindent{\bf Lemma~#1~} \em }{\vspace{0.1in}}
\newenvironment{appendix-theo}[1]{\vspace{0.1in}\noindent{\bf Theorem~#1~} \em }{\vspace{0.1in}}
\renewenvironment{proof}{{\textsc{Proof:}}}{\qed}
\newcommand{\proofofref}{}
\newproof{zproofof}{Proof of \proofofref}
\newenvironment{proofof}[1] {\renewcommand{\proofofref}{#1}\zproofof}
 {\endzproofof}

\usepackage[margin = 1in]{geometry}

\title{Popular Critical Matchings in the Many-to-Many Setting\footnotetext{A preliminary version~\cite{NNRS21} of this work appeared in 41st IARCS Annual Conference on Foundations of Software Technology and Theoretical Computer Science (FSTTCS 2021)}} 

\journal{Theoretical Computer Science}

\begin{document}

\begin{frontmatter}



\author[1]{Meghana Nasre}
\ead{meghana@cse.iitm.ac.in}
\affiliation[1]{organization={Department of CSE},
            addressline={IIT Madras},
            city={Chennai},
            postcode={600036}, 
            state={Tamilnadu},
            country={India}}
\author[2]{Prajakta Nimbhorkar}
\ead{prajakta.nimbhorkar@gmail.com}

\affiliation[2]{organization={Department of CSE},
            addressline={Chennai Mathematical Institute and UMI ReLaX},
            city={Chennai},
             postcode={603103}, 
            state={Tamilnadu},
            country={India}}
 \author[1]{Keshav Ranjan}  
\ead{keshav@cse.iitm.ac.in}

\author[3,fn3]{Ankita Sarkar}
\ead{Ankita.Sarkar.GR@dartmouth.edu}
\affiliation[3]{organization={Department of Computer Science},
            addressline={Dartmouth College},
            city={Hanover},
            state={NH 03755},
            country={USA}}
  \fntext[fn3]{The work was partly done when the author was a Master's student at Chennai Mathematical Institute.}         



\begin{abstract}
We consider the many-to-many bipartite matching problem in the presence of two-sided preferences and two-sided lower quotas. The input to our problem is a bipartite graph $G=(\AA \cup \BB, E)$, where  each vertex in $\AA \cup \BB$ specifies a strict preference ordering over its neighbors. Each vertex has an upper quota and a lower quota denoting the maximum and minimum number of vertices that can be assigned to it from its neighborhood.
In the many-to-many setting with two-sided lower quotas, informally, a \emph{critical} matching is a matching which fulfils vertex lower quotas to the maximum possible extent. This is a natural generalization of the definition of a critical matching in the one-to-one setting~\cite{Kavitha2021}.
Our goal in the given problem is to find a \emph{popular} matching in the set of critical matchings. A matching is popular in a given set of matchings if it remains undefeated in a head-to-head election with any matching in that set. Here, vertices cast votes between pairs of matchings. We show that there always exists a matching that is popular in the set of critical matchings. We present an  efficient algorithm to compute such a matching of the largest size. We prove the popularity of our matching using a dual certificate.

\end{abstract}


\end{frontmatter}

\input{1short_intro}

\input{2algo}
\input{3cloneG}
\input{4crit}
\input{5popular}

\input{6conclusion}

 \bibliographystyle{elsarticle-num-names} 
 \bibliography{refs.bib}
\appendix
\input{8appendix}
\end{document}

%% file: 1short_intro.tex
\section{Introduction}\label{sec:intro}
In this paper, we study the many-to-many bipartite matching problem in the presence of two-sided preferences where vertices in both partitions specify lower and upper quotas. The many-to-many matching problem models applications like assigning workers to firms~\cite{roth1984stability} and students to courses~\cite{CEFMMP2014} where both sides of the bipartition can accept multiple partners. There also exist markets which are typically many-to-one, but some agents in those markets are multiple job holders. For instance -- around 5\% employees in the U. S. are multiple jobholders\footnote[3]{Source: Annual averages of employed multiple job holders by industry, Division of Labor Force Statistics, U.S. Bureau of Labour Statistics.} and around 35\% of teachers in Argentina work in two or more schools~\cite{echenique2006theory}. It is natural to have \emph{preferences} on both sides of the bipartition in matching applications like student-course allocation\footnote[4]{In literature, the student-course allocation problem has also been considered as a one-sided matching problem where courses are treated as objects. For instance -- in~\cite{CEFMMP2014}, authors deal with the applicant-course allocation problem as one-sided matching markets where the notion of optimality considered is pareto optimality. In general, the courses offered at a university have bounded quotas to prevent the situation where some courses are overloaded. Once quotas are full, courses have to decide which students' application to discard. One way of limiting the students' application is to use course preferences -- making the problem two-sided. The student-course allocation problem has also been studied as a two-sided matching market in the literature. For example -- in~\cite{brandl2019two}, authors study it as a many-to-many matching problem with two-sided preferences without lower quotas where the notion of optimality considered is popularity.} or worker-firm assignment. For example -- consider the student-course allocation problem where students have preferences over various courses based on their interests, and a course-teacher specifies ordering of students based on their performances in the corresponding prerequisite courses. In the student-course allocation problem, students typically have a \emph{minimum} requirement on the number of courses they need to complete in a semester, and a course may be offered only if there is a \emph{minimum} number of registrants. Similarly, in the worker-firm assignment, certain workers \emph{must} be assigned jobs, and firms may need a \emph{minimum} number of workers for their operations. Lower quotas allow us to capture these constraints. These matching markets require a systematic study of many-to-many two-sided matching markets with two-sided lower quotas.

The input to our problem is a bipartite graph $G=(\AA\cup\BB,E)$, where $\AA$ and $\BB$ are two sets of vertices and $E$ denotes the set of all the acceptable vertex-pairs. Every vertex $u\in \AA\cup\BB$ has a strict preference ordering on its neighbours, called the {\em preference list} of $u$, denoted as $\prefu$. 
Associated with every vertex $u\in\AA\cup\BB$ is a lower quota $q^-(u)\in \mathbb{Z}^+\cup\{0\}$ and  an upper quota $q^+(u)\in \mathbb{Z}^+$ such that $q^-(u)\le q^+(u)$. The lower quota $q^-(u)$ denotes the minimum number of vertices from its neighbourhood that {\em must be} assigned to $u$ and the upper quota $q^+(u)$ denotes the maximum number of vertices that $u$ can accommodate in any assignment. If $q^-(u)>0$ for a vertex $u$ then we call $u$ an \emph{lq-vertex}. A \emph{matching} $M$ in $G$ is a subset of the edge set $E$ such that $|M(u)|\le q^+(u)$ for each vertex $u\in\AA\cup\BB$, where $M(u)$ denotes the set of neighbours assigned to $u$ in $M$. Note that this definition of matching deviates from the classical one used in graph theory. We still use the term \emph{matching} as done in the literature on matchings under preferences \cite{brandl2019two,HIM16}. In the matching $M$, a vertex $u\in\AA\cup\BB$ is called \emph{fully-subscribed} if $|M(u)| = q^+(u)$, \emph{under-subscribed} if $|M(u)| < q^+(u)$, \emph{deficient} if $|M(u)| < q^-(u)$, and \emph{surplus} if $|M(u)| > q^-(u)$.  A matching $M$ is {\em feasible} if no vertex in $\AA \cup \BB$ is deficient in $M$. Feasible matchings are desirable in the sense that they always fulfil the demand of every vertex. Unfortunately, the existence of a feasible matching is not guaranteed. 

Consider an instance with $\AA=\{a_1,a_2,a_3\}$ and $\BB=\{b_1,b_2\}$ shown in Figure~\ref{subfig:exCrit}. In this instance, the sum of upper quotas of vertices on the $\BB$-side is three whereas the sum of lower quotas of vertices on the $\AA$-side is four. Thus, a feasible matching does not exist for this instance. In such a situation, we seek to compute a matching that is as close as possible to a feasible matching. Informally, a \emph{critical} matching~\cite{Kavitha2021} is a matching which fulfils vertex lower quotas to the maximum possible extent. Thus, if the instance admits a feasible matching, then the set of critical matchings and the set of feasible matchings are exactly the same. In this work, we are interested in computing a critical matching that is {\em optimal} with respect to the preferences of the vertices. 

For a matching $M$, we define the deficiency of a vertex $u\in\AA\cup\BB$  as $\max\{0,q^-(u)-|M(u)|\}$. The deficiency, $\Df{M}$, of a matching $M$ is equal to the sum of the deficiencies of all the vertices $u\in\AA\cup\BB$ in $G$. A critical matching is a matching with minimum deficiency among all the matchings. Given a matching $M$, let $\Dfa{M}$ and $\Dfb{M}$  denote the sum of deficiencies of all the vertices in $\AA$ and all the vertices in $\BB$, respectively, in $M$. That is, 

$$\Dfa{M}=\sum_{a\in\AA\ and\ |M(a)|<q^-(a)} ( q^-(a)-|M(a)|)$$ and  $$\Dfb{M}=\sum_{b\in\BB\ and\ |M(b)|<q^-(b)} (q^-(b)-|M(b)|)$$
In other words, $\Dfa{M}$ and $\Dfb{M}$ denote the total number of lower-quota positions that remain unfilled in $M$ for all the vertices in $\AA$  and all the vertices in $\BB$, respectively. Note that $\Df{M}=\Dfa{M}+\Dfb{M}$. 

\begin{definition}[Critical Matching]\label{def:critical}
A matching $M$ in $G$ is critical if there is no matching $N$ in $G$ such that $\Df{N}<\Df{M}$. 
\end{definition}

 For a critical matching $M$ and any other matching $N$ (not necessarily critical) it is easy to show that $\Dfa{M}\le\Dfa{N}$ and $\Dfb{M}\le\Dfb{N}$  (see Claim~\ref{lem:noLessDef}). In the example shown in Figure~\ref{fig:exCritical} consider the three matchings $M,M'$ and $N$ shown in Figure~\ref{subfig:matchExCrit}. The deficiencies of these matchings are $\Df{M'}=2, \Df{M}=\Df{N}=1$. Since $G$ does not admit a feasible matching, $M$ and $N$ are critical matchings in $G$.

\begin{figure}
\centering
\begin{subfigure}{.5\textwidth}
  \begin{tikzpicture}[thick,
  lnode/.style={draw=white,fill=white},
  fsnode/.style={draw,circle,fill=black,scale=0.5},
  every fit/.style={ellipse,draw,inner sep=-2pt,text width=0.5cm},
  -,shorten >= 3pt,shorten <= 3pt
]

  \node[fsnode] (a1) at (0,3) {};
  \node[fsnode] (a2) at (0,2.2) {};
  \node[fsnode] (a3) at (0,1.4) {};
  \node[fsnode] (b1) at (2.5,3) {};
  \node[fsnode] (b2) at (2.5,1.9) {};
\node at (-1.1,3) {$[1,2]\ a_1$};
\node at (-1.1,2.2) {$[2,2]\ a_2$};
\node at (-1.1,1.4) {$[1,1]\ a_3$};
\node at (3.5,3) {$b_1\ [0,1]$};
\node at (3.5,1.5) {$b_2\ [1,2]$};

\node [fit=(a1) (a3),label=above:$\AA$] {};
\node [fit=(b1) (b2),label=above:$\BB$] {};

\draw(a1) -- (b1) node[pos=0.3,draw=white, fill=white,inner sep=0.1pt] {\textcolor{black}{\small 1}} node[pos=0.7,draw=white, fill=white,inner sep=0.1pt] {\textcolor{black}{\small 1}};
\draw(a1) -- (b2) node[pos=0.2,draw=white,fill=white,inner sep=0.1pt]{\textcolor{black}{\small 2}} node[pos=0.7,draw=white, fill=white,inner sep=0.1pt] {\textcolor{black}{\small 2}};
\draw(a2) -- (b1) node[pos=0.3,draw=white,fill=white,inner sep=0.1pt]{\textcolor{black}{\small 1}} node[pos=0.8,draw=white, fill=white,inner sep=0.1pt] {\textcolor{black}{\small 2}};
\draw(a2) -- (b2) node[pos=0.2,draw=white,fill=white,inner sep=0.1pt]{\textcolor{black}{\small 2}} node[pos=0.78,draw=white, fill=white,inner sep=0.1pt] {\textcolor{black}{\small 3}};
\draw(a3) -- (b2) node[pos=0.25,draw=white,fill=white,inner sep=0.1pt]{\textcolor{black}{\small 1}} node[pos=0.65,draw=white, fill=white,inner sep=0.1pt] {\textcolor{black}{\small 1}};
\end{tikzpicture}
\caption{}
  \label{subfig:exCrit}
\end{subfigure}%
\begin{subfigure}{.5\textwidth}
  \centering
  \begin{tabular}{cl}
            $M=$&$\{(a_1,b_1),(a_2,b_2),(a_3,b_2)\}$\\
            $M'=$ &$\{(a_1,b_1),(a_1,b_2),(a_3,b_2)\}$\\
            $N=$&$\{(a_1,b_2),(a_2,b_1),(a_2,b_2)\}$
		\end{tabular}
  \caption{}
  \label{subfig:matchExCrit}
\end{subfigure}
\caption{\label{fig:exCritical}%
(i) Example instance -- The numbers in [ ] denote the lower quota and upper quota, respectively, for the corresponding vertex. The numbers on the edges denote the ranking of the other endpoint in the preference list of that vertex. (ii) Three different matchings $M, M'$ and $N$. The deficiencies of these matchings are $1,2$ and $1$, respectively.
}
\end{figure}
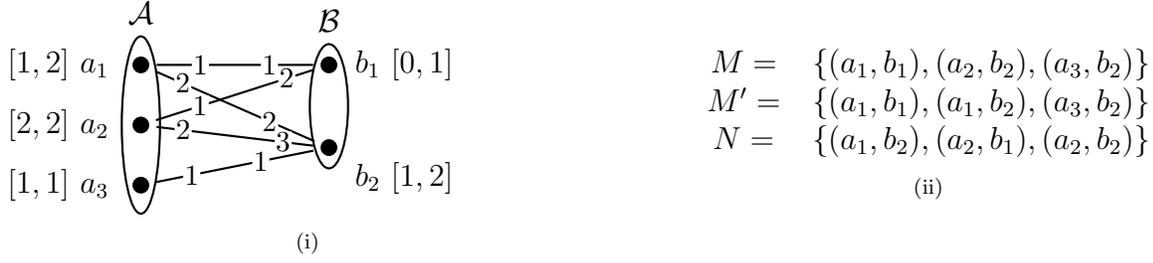


Our setting is a generalization of the one-to-one stable marriage problem and many-to-one hospital residents ({\sf HR}) problem without lower quotas, both of which have been extensively investigated \cite{GS62,roth84,Roth86} and recently in the presence of lower quotas~\cite{NN17,NNRS21,Kavitha2021}.
The classical notion of optimality for the matching problem with two-sided preferences is \emph{pairwise stability}, henceforth called stability and is characterized by the absence of a \emph{blocking pair}. A pair $(a, b) \in E \setminus M$ is called a \emph{blocking pair} with respect to the matching $M$ if $(i)$ either $|M(a)| < q^+(a)$ or $b$ precedes some $b' \in M(a)$ in $\prefa$ and $(ii)$ either $|M(b)| < q^+(b)$ or $a$ precedes some $a'\in M(b)$ in $\prefb$.

The well-known Gale and Shapley algorithm~\cite{GS62} for the one-to-one setting can be suitably modified to compute a stable matching in a many-to-many instance without lower quotas. A stable matching always exists; however, in the presence of lower quotas, it may not be critical -- this holds even in the one-to-one setting, with lower quotas only on one side of the bipartition. Using the Rural Hospitals theorem~\cite{Roth86}, it is possible to determine whether the instance admits a stable and critical matching.
 Since stability and criticality are not simultaneously guaranteed, an alternate notion of optimality, namely \emph{popularity} \cite{G75} is used in the literature~\cite{NN17,NNRS21,Kavitha2021}. Informally, a matching is \emph{popular} in a given set of matchings if it remains undefeated in a head-to-head election with any matching in that set. The notion of popularity provides an \emph{overall} stability since a majority of vertices do not wish to deviate from a popular matching.

Before stating our results, we formally define the notion of popularity below.

\noindent \textbf{Popularity:} The notion of popularity used in this work is a generalization of the one used in~\cite{NN17,NNRS21}. Each vertex casts votes to compare two given matchings $M$ and $N$. A vertex $v\in\AA\cup\BB$ can get matched to at most $q^+(v)$ many vertices in $M$ and $N$. If $v$ is under-subscribed in $M$ or $N$, then we assume that the remaining positions of $v$ are matched to appropriately many instances of $\bot$.  The vertex $v$ prefers any vertex in $\prefv$ over $\bot$ and is indifferent among all $\bot$ matched to it. Thus, we can always assume that $|M(v)|=|N(v)|=q^+(v)$ and $v$ can cast up to $q^+(v)$ votes. The vertex $v$ is indifferent between the two matchings $M$ and $N$ for $|M(v)\cap N(v)|$ many positions and does not cast votes for these positions.  Note that $|M(v)\setminus N(v)|= |N(v)\setminus M(v)|$ and hence, $v$ decides a correspondence function $\cor_v$ to compare a vertex $u\in M(v)\setminus N(v)$ and $u'\in N(v)\setminus M(v)$. The correspondence function $\cor_v$ is an arbitrary bijection between the two sets $M(v)\setminus N(v)$ and $N(v)\setminus M(v)$. That is, for each vertex $u\in M(v)\setminus N(v)$, $\cor_v(u,M,N)$ gives a vertex $u'\in N(v)\setminus M(v)$ and vice versa. Thus, $vote_v(u,\cor_v(u,M,N))$ is 1 if $v$ prefers $u$ over $u'$, and $-1$ if $v$ prefers $u'$ over $u$. 
To compare matchings $M$ and $N$, each vertex fixes $\cor_v$. Now, we count the total votes by $v$ when comparing $M$ and $N$ as 

 $$vote_v(M,N,\cor_v) = \sum\nolimits_{u\in M(v)\setminus N(v)} vote_v(u,\cor_v(u,M,N)) $$
 
Thus, the number of votes that $M$ gets over $N$ is given by 

$$\Delta(M,N,\cor)=\sum\nolimits_{v\in \AA\cup\BB}vote_v(M,N,\cor_v) $$ where $\cor$ is called the correspondence function from $M$ to $N$ and denotes the \emph{disjoint union} of $\cor_v$ for all $v\in\AA\cup\BB$.

\begin{definition}[Popular Matching]\label{def:popular}
A matching $M$ is more popular than $N$ under $\cor$ if $\Delta(M,N,\cor)>0$. A matching $M$ is called popular if there is no matching $N$ such that $N$ is more popular than $M$ for any choice of $\cor$.
\end{definition}

The number of votes cast by a vertex $v$ to compare two matchings $M$ and $N$ depends on the correspondence function $\cor_v$ fixed by the vertex $v$. Let us consider an instance $G=(\AA\cup\BB,E)$ where $\AA=\{a_1,a_2,a_3,a_4,a_5,a_6\}, \BB=\{b\}, q^+(a)=1$ for all $a\in\AA$, $q^+(b)=3; q^-(v)=0$ for all $v\in\AA\cup\BB$, and $E$ contains the edges of a complete bipartite graph between $\AA$ and $\BB$. The preference list of $b$ is defined as $\prefb : a_1,a_2,a_3,a_4,a_5,a_6$. Consider two matchings $M=\{(a_2,b),(a_3,b),(a_5,b)\}$ and $N=\{(a_1,b), (a_4,b),(a_6,b)\}$. The vertex $b$ compares the two sets $\{a_2,a_3,a_5\}$ and $\{a_1,a_4,a_6\}$ using the correspondence function $\cor_{b}$. If $b$ fixes $\cor^1_{b}$ such that $a_2,a_3, a_5$ are mapped to $a_1,a_4,a_6$ respectively, then $vote_{b}(M,N,\cor^1_{b})=-1+1+1=1$. Whereas, if $b$ fixes $\cor^2_{b}$ such that $a_2,a_5, a_3$ are mapped to $a_1,a_4,a_6$ respectively, then $vote_{b}(M,N,\cor^2_{b})=-1-1+1=-1$. 

Note that the votes by vertices in $\AA$ while comparing $M$ and $N$ cancel out as the set of matched vertices from $\AA$-side in both the matchings are disjoint. This implies that $M$ is more popular than $N$ under $\cor^1_{b}$, and $N$ is more popular than $M$ under $\cor^2_{b}$. Thus, the popularity of $M$ compared to $N$ or vice versa crucially depends on the correspondence function $\cor$. Yet, interestingly, our algorithm outputs a matching that is popular for \emph{any} choice of $\cor$. For the above instance $G$, our algorithm outputs the matching $M'=\{(a_1,b),(a_2,b),(a_3,b)\}$. It can be verified that there does not exist any correspondence function under which $M'$ is beaten by some other critical matching.

\vspace{0.1in}

\noindent{\bf Our results:}\label{sec:results}
 The main contribution of this work is an efficient algorithm to compute a popular matching in the set of critical matchings, henceforth called popular critical matching, for the many-to-many setting with two-sided preferences and two-sided lower quotas.
\begin{theorem}\label{theo:main}
  Let $G=(\AA\cup\BB,E)$ be a bipartite graph where each vertex $v\in\AA\cup\BB$ has an associated lower quota $q^-(v)$, an upper quota $q^+(v)$ and a strict preference ordering over its neighbours. Then, there exists a polynomial time algorithm to compute a maximum size matching $M$ that is popular in the set of critical matchings of $G$.  
\end{theorem}

Recall the many-to-one setting without lower quotas \emph{i.e.} the \HR\ problem. For any instance of the \HR\ problem, the Rural Hospitals theorem~\cite{Roth86} states that every stable matching matches each hospital to the same number of residents and all the under-subscribed hospitals to the \emph{same set} of residents. 
A modified version of the Rural Hospitals theorem holds for popular matchings in the many-to-many setting without lower quotas~\cite{brandl2019two}. We show a similar version of the Rural Hospitals theorem for our setting:

\begin{theorem}\label{theo:ruralHos}
Let $M$ be the matching computed by our algorithm for a bipartite graph $G=(\AA\cup\BB,E)$ where each vertex $v\in\AA\cup\BB$ has an associated lower quota $q^-(v)$, an upper quota $q^+(v)$ and a strict preference ordering over its neighbours. Then every max-size popular critical matching of $G$ matches the same set of vertices as $ M$, and each vertex $v\in \AA\cup\BB$ matches to exactly $|M(v)|$ many vertices in every other max-size popular critical matching.
\end{theorem}

Next, we show that although we are restricting ourselves to the set of critical and popular matchings, we are not compromising too much on the size when compared to an arbitrary \emph{maximum size} matching. A similar result was shown in~\cite{brandl2019two} for popular matchings in the many-to-many setting without lower quotas.

\begin{theorem}\label{theo:2by3}
Let $M$ be the matching computed by our algorithm and $M_{max}$ be an arbitrary maximum size matching for a bipartite graph $G=(\AA\cup\BB,E)$ where each vertex $v\in\AA\cup\BB$ has an associated lower quota $q^-(v)$, an upper quota $q^+(v)$ and a strict preference ordering over its neighbours. Then $|M|\ge \frac{2}{3}|M_{max}|$.
\end{theorem}

\noindent\textbf{Challenges in many-to-many setting:} In general, it is nontrivial to adapt one-to-one or many-to-one matching algorithms to the many-to-many setting. This is pointed out in the literature (see \cite{kamiyama2019many,chen2010strongly}). Specifically in our case, there are challenges in extending the approaches considered earlier for computing popular feasible matching in the one-to-one and many-to-one settings to the many-to-many setting. The algorithms in \cite{NNRS21,Kavitha2021} both propose a reduction where the original instance with two-sided lower quotas is converted into an instance {\em without} lower quotas. The standard Gale-Shapley algorithm~\cite{GS62} is used to obtain a  stable matching in the reduced instance, and the edges in the stable matching are mapped to the edges in the original instance. The reduced instance constructed in \cite{NNRS21,Kavitha2021}  has  multiple {\em level} copies of the original vertices. The notion of levels for vertices will become clear when we give the description of our algorithm in Section~\ref{sec:Algo}.
Extending a similar approach to the many-to-many setting can lead to a vertex $a$ being matched to $b$ multiple times in the reduced instance, and such matching cannot be mapped back to a feasible/critical matching in the original instance. Such a difficulty does not arise when the quota of at least one side of the bipartition is one which is indeed the case in the earlier works ~\cite{NN17,NNRS21,Kavitha2021} (see \ref{sec:challenges} for an example). In this work, we present a proposal-based algorithm inspired by a similar idea mentioned by Kavitha in~\cite{Kavitha2021}.

\noindent {\bf Related work: }\label{sec:rel-work}
The notion of popularity was introduced by {G{\"a}rdenfors}~\cite{G75}  as  \emph{majority assignment} for the stable marriage problem with \emph{ties}.  Abraham~\etal~\cite{abraham2007popular} gave an efficient
 algorithm for computing popular matching in a bipartite graph where agents in only one side of the bipartition have preferences. Subsequently,  popularity is investigated in the two-sided world~\cite{biro2010popular,huang2013popular,kavitha2014size,hirakawa2015structure,cseh2018popular,MNNR18,kavitha2016popular}. This notion of optimality  has been investigated in various settings of the two-sided preference list models like -- \emph{with} lower quotas~\cite{NN17,MNNR18,Kavitha2021,NNRS21}  and \emph{without} lower quotas~\cite{huang2013popular,kavitha2014size,nasre2017popularity,brandl2019two}.  In the two-sided settings \emph{without} lower quotas popularity has been considered as an alternative to stability. For instance, in \cite{huang2013popular,kavitha2014size}, popularity was proposed for a stable marriage problem to obtain optimal matchings larger in size than a stable matching. 
Unlike stable matchings where all of them have the same size, popular matchings can be of varying sizes. In a stable marriage instance (one-to-one setting), a maximum cardinality popular matching is linear-time computable and popular matching amongst maximum cardinality matchings is polynomial-time computable~\cite{kavitha2014size}.  Popularity in the generalized \HR\ setting (Laminar Classified Stable Matching problem) and many-to-many setting has been considered in~\cite{nasre2017popularity} and~\cite{brandl2019two} respectively, both without lower quotas.  Readers are referred to the survey~\cite{cseh2017popular}  by {\'A}gnes Cseh for recent results about popularity in different settings.
There has been a lot of follow-up work on popular matchings in bipartite one-to-one, many-to-one, and many-to-many settings, their properties and generalizations like min-cost popular matchings ~\cite{Kavitha20}, dual certificates to popularity ~\cite{kavitha2016popular,KKMSS20}, popular matchings polytope and on their properties ~\cite{HK21,FK20}. Other follow-up works on popular matchings include quasi-popular matchings~\cite{FK20}, popular max-matching polytope~\cite{kavithaICALP2021}, dominant matchings polytope \cite{cseh2018popular,FK20}, popular matchings with one-sided bias \cite{kavithaOnesideBiasPop}, popular matchings of desired size~\cite{Kavitha18}, popular matchings with one-sided ties~\cite{CHK17,GNNR19}, complexity-theoretic results about popular and dominant matchings \cite{AFKP22}.

Now we restrict our attention to works where lower quotas are imposed. A special case of our problem called popular matching amongst \emph{feasible} matchings in the presence of lower quotas has been considered in~\cite{NN17} for the \HRLQ\ setting, where only hospitals have lower quotas. Nasre and Nimbhorkar~\cite{NN17}  showed that the maximum cardinality popular matching amongst feasible matchings for an \HRLQ\ problem is efficiently computable. An empirical analysis of popular feasible matchings and \emph{envy-free} matchings for the \HRLQ\ problem has been done in~\cite{MNNR18} where authors observed that popular matchings outperform {\em envy-free} matchings on parameters of practical relevance. Very recently, Kavitha~\cite{Kavitha2021} explored  popularity for one-to-one setting in the two-sided preference list model with two-sided lower quotas and showed that a maximum size popular matching amongst all critical matchings in this setting is efficiently computable. Independent of the work in~\cite{Kavitha2021} and concurrently, we~\cite{NNRS21} generalized the work in~\cite{NN17} to the case of \HRTLQ\ i.e. many-to-one setting in two-sided preference list model with two-sided lower quotas. In this work, we nail down the final piece of the problem for popularity in the many-to-many setting with lower quotas. We show that maximum size popular matching amongst all critical matchings for many-to-many setting in the two-sided preference list model with two-sided lower quotas is efficiently computable. 

As pointed out earlier, a critical stable matching may not exist for an instance with lower quotas. Thus, to deal with lower quotas, different notions of optimality other than popularity, have been considered in the literature. Hamada \etal\ in~\cite{HIM16} considered the problem of computing a feasible matching with \emph{minimum} number of blocking pairs or blocking residents. They showed that both of these problems are NP-Hard even under severe restrictions on preference lists. Fleiner and Kamiyama~\cite{FK16} investigated two-sided lower quotas in the many-to-many setting. In their model, they allow classifications as well. However, their goal was to decide whether the instance admits a stable matching or not. \emph{Envy-freeness}~\cite{wu2018lattice} is another well-investigated optimality notion in \HR\ settings. This is a relaxed notion of stability which forbids the existence of a resident who has justified envy towards another resident. Yokoi in~\cite{Yokoi20} investigated envy-freeness for \HRLQ\ setting and provided a characterization for \HRLQ\ instances that admit an envy-free matching. She also gave a linear-time algorithm to compute an envy-free matching, if it exists. Krishnaa \etal~\cite{girija2020envy} investigated the complexity of computing maximum size envy-free matchings and also introduced and studied a new notion of optimality called \emph{relaxed stability} for the \HRLQ\ setting.  Mnich and Schlotter in~\cite{mnich20} considered  stable marriage problem (one-to-one) with two-sided lower quotas under a relaxed notion of tractability, namely fixed-parameter tractability. There they studied how a set of natural parameters like maximum length of a preference list, allowed number of blocking pairs and total number of lower quota vertices determine the computational tractability of this problem. Recently, Goko \etal\ in~\cite{goko2022maximally} considered \HR\ problem in the presence of one-sided lower quotas and two-sided ties. They studied the problem of computing weakly stable matching with maximum total \emph{satisfaction ratio} for lower quotas, where satisfaction ratio of each hospital $h$ is given by $\min \{1,\frac{|M(h)|}{q^-(h)}\}$.

%% file: 2algo.tex
\section{Our algorithm}
\label{sec:Algo}
We first give a high-level idea of our algorithm. As mentioned earlier, our algorithm is a proposal-based algorithm. For the many-to-many setting, in the absence of lower quotas on any of the sides, a proposal-based algorithm is presented in~\cite{brandl2019two} which works as follows. One of the sides of the bipartition, say the $\AA$ side, proposes to vertices on the other, and the standard Gale and Shapley algorithm~\cite{GS62} is executed. At this stage, all the proposing vertices are considered to be at the lowest level, i.e. level zero. If at the end of this proposal sequence, some vertices are under-subscribed, they are given another \emph{chance}, that is, the level of the vertex is increased, and the vertices are allowed to propose again. A vertex $b$ which receives the proposal always prefers a higher level vertex  $a$ over any lower level vertex $a'$  irrespective of the preference of $b$ between $a$ and $a'$.

If there were lower quotas for vertices on the $\AA$ side alone, an extension of the above idea as given in \cite{MNNR18} is to let deficient vertices on the $\AA$ side propose with even higher levels.  It can be shown that the number of \emph{additional} levels required to compute a popular feasible/critical matching is equal to the sum of lower quotas of the vertices on the $\AA$ side. 
Furthermore, since vertices on the $\AA$ side are proposing to satisfy their deficiency, at a higher level ($>1$), the vertex is allowed to get matched to at most lower quota many partners.
In our case, we have lower quotas on both sides. Thus before we start our standard Gale and Shapley algorithm, we ensure that the deficiency of the $\BB$ side is fulfilled. In order to do so, we let vertices on the $\AA$ side propose to {\em only} lower quota vertices on the $\BB$ side for a certain number of levels {\em before} the Gale and Shapley proposals begin. Since the goal is to meet the deficiency of $\BB$ side vertices, in a {\em lower level}, vertex $b$ is allowed to be matched to at most lower quota many partners. There are additional subtleties in making this overall idea work which we elaborate below formally.

Let $\S$ and $\T$ denote the sum of lower quotas of all the vertices in $\AA$ and $\BB$, respectively. {That is, $\S=\sum_{a\in\AA}q^-(a)$ and $\T = \sum_{b\in\BB}q^-(b) $.}  The algorithm uses $\S+\T+2$ levels $0,1,\ldots, \T, \T+1, \ldots, \S+\T+1$. All vertices in $\AA$ 
begin at level $0$. During the algorithm, a vertex $a\in\AA$ raises its level multiple times, possibly up to the highest level, that is, $\S+\T+1$. A vertex $a$ at level $\ell$ is denoted as $a^\ell$. A vertex $b\in\BB$ prefers $a_i^\ell$  over $a_j^{\ell'}$ if : 
\begin{itemize}
    \item[(i)] either $\ell > \ell'$ (relative positions of $a_i$ and $a_j$ in $\prefb$ do not matter) or
    \item[(ii)] $\ell = \ell'$ and $a_i$ precedes $a_j$ in $\prefb$.
\end{itemize} 

Let $\preflqa$ denote the preference list of $a$ restricted to \emph{lq}-vertices in its preference list $\prefa$. For instance if $\prefa = b_1, b_2, b_3, b_4$ where $b_2$ and $b_4$ are \emph{lq}-vertices, then $\preflqa = b_2, b_4$.  During the course of the algorithm, the capacity of the proposing vertex $a$, denoted by $c(a)$, and the set of neighbours that it proposes to depend on the level of the vertex (see Table~\ref{tab:quota}).
This is a crucial aspect of our algorithm, ensuring criticality on both sides as well as the popularity of our matching.

\begin{table}
\centering
\begin{tabular}{| p{3 cm} |p{1.1cm} | p{2cm} |p{1cm} p{4cm}|}
\hline
      \textbf{level of $a$}  & \textbf{$c(a)$} & \textbf{preference list of $a$ }&  & \textbf{$\qquad \ \ c(b)$}  \\
      \hline
      \hline 
      $0, \ldots, \T  - 1$ &  & \preflqa & $q^-(b)$ & \\ [2pt]
      \cline{1-1} \cline{3-5} 
      $\T$, $\T + 1$ & $q^+(a)$ & \prefa & $q^-(b)$ & if $\exists\ a_i\in M(b)$ such that $a_i$ is at level $< \T$ \\ [2pt]
      \cline{1-2}
      $\T+2, \ldots, \S+\T+1$ & $q^-(a)$ &  & $q^+(b)$ & otherwise \\[2pt]
\hline     
\hline
\end{tabular}
\vspace{0.1in}
\caption{Let $a \in \AA$ be the vertex proposing to $b \in \BB$. Entries in the table give the capacity and preference list of vertex $a$ and the capacity of the vertex $b$ used by Algorithm~\ref{algo:maxPopSC2LQ}. We denote the capacity of a vertex $v$ by $c(v)$.}
\label{tab:quota}
\end{table}

 Initially, each $a\in\AA$ is at level 0 and starts proposing to vertices in $\preflqa$ with its capacity, $c(a)$, set to $q^+(a)$. If  $|M(a)|<c(a)$ after completing its proposals at a level $0 \le \ell < \T - 1$, then $a$ raises its level by one and continues proposing to vertices in $\preflqa$.
 Once the vertex $a$ is at level $\T$ or higher, it proposes to {\em all} neighbours in its preference list, that is, to vertices in $\prefa$. If $|M(a)|<c(a)$ at level $\T$, then it raises its level by one and continues the proposals. If $a$ remains deficient after completing its proposals at level  $\T+1 \le \ell \le \S + \T$, then it raises its level by one and proposes to vertices in $\prefa$ with the capacity $c(a)$ set to $q^-(a)$.

The capacity, $c(b)$, of a vertex $b$ receiving the proposal is determined by the level of the proposing vertex and by the level of its matched partners at that time.
If $b$ receives a proposal from $a$ at level $0 \le \ell \le \T -1 $, then the capacity of $b$ is $q^-(b)$. On the other hand, if the level of the proposing vertex is $\T \le \ell \le \S+\T+1$, then the capacity of $b$ is determined as follows:
\begin{itemize}
    \item If there exists at least one vertex in $M(b)$ at level $\T-1$ or lower, then the capacity of $b$ is $q^-(b)$.
    \item Otherwise the capacity of $b$ is $q^+(b)$.
\end{itemize}
Whenever $b$ is full with respect to its capacity, and it receives a proposal from $a$, it rejects the least preferred vertex in $M(b) \cup \{a\}$, taking into consideration the levels of the vertices.

Throughout the algorithm, once a vertex $a$ raises its level, it starts proposing to vertices from the beginning of its current preference list ($\prefa$ or $\preflqa$) until it fulfils its capacity. If a vertex $b$ is already matched to $a$ at a lower level, then $b$ rejects the previous proposal and accepts the new one at the current level.

The approach is formally described in Algorithm~\ref{algo:maxPopSC2LQ}.  The output matching of this algorithm is $M$. 
Algorithm~\ref{algo:maxPopSC2LQ} uses a procedure \texttt{DecideAccRej} that decides whether a vertex $b\in\BB$ accepts/rejects the proposal from $a\in\AA$ based on the current matching and capacity of $b$. Lines~\ref{dec:art1} and \ref{dec:art2} in the \texttt{DecideAccRej} ensure that at most one level copy of a vertex $a$ is matched to $b$ throughout the algorithm. Lines~\ref{dec:st1}-\ref{dec:st2} are the standard accept/reject procedure. Lines~\ref{dec:que1} and~\ref{dec:que2} add the vertex $a$ at  level $\ell$ to the queue $Q$ provided $a$ is not full with respect to its capacity and 
no other level copy of $a$ is present in $Q$.

\begin{algorithm}
    \SetKwFunction{decision}{DecideAccRej}
    \caption{Max-size popular critical matching in $G = (\AA \cup \BB, E)$ }\label{algo:maxPopSC2LQ}
    \DontPrintSemicolon
    \SetAlgoLined
    \begin{spacing}{1.2}
    Let $\S=\sum_{a\in\AA}q^-(a)$ and $\T = \sum_{b\in\BB}q^-(b) $ \;
    \end{spacing}
    Set $M=\emptyset$, Initialize a queue $Q=\{a^0\ :\ a\in \AA\}$\;
    \While{$Q$ is not empty}{
        Let $a^\ell=dequeue(Q)$\;
        \If {$\ell<\T$ \label{alg:all2lq} \tcc*{ $|M(a)|<q^+(a)$} } { 
            \If{ $a^\ell$ has not exhausted $\preflqa$}{Let $b=$  most preferred unproposed vertex by $a^\ell$ in $\preflqa$ \;
            \decision{$a^\ell,q^+(a),b,q^-(b)$}\label{alg:1stdecide}\;
            }
            \Else{$\ell=\ell+1$ and add $a^\ell$ to $Q$\;}
        }
        \Else{
        \textbf{if} $\ell==\T$ or $\ell==\T+1$ \textbf{then} $c(a) = q^+(a)$
            \textbf{else} $c(a) = q^-(a)$ \label{alg:LQstudents} \;
            \If{$a^\ell$ has not exhausted $\prefa$}{Let $b=$ most preferred unproposed vertex by $a^\ell$ in $\prefa$\;
                \If{$|M(b)|<q^-(b)$}{\decision{$a^\ell,c(a),b,q^-(b)$}\;}
                \uElseIf{$|M(b)|==q^-(b)$}{ 
                \If{there exists a vertex in $ M(b)$ at level less than $\T$ \label{alg:lowerLevel}}{
                \decision{$a^\ell,c(a),b,q^-(b)$} \label{alg:lowLDecide}\;
                   } \Else{\decision{$a^\ell,c(a),b,q^+(b)$}\label{alg:22} \; }}                \Else{\decision{$a^\ell,c(a),b,q^+(b)$}\label{alg:24}\;}}
                \Else{ 
                \If{$(\ell<\S+\T+1$ and $|M(a)|<q^-(a))$ or $(\ell==\T)$ \label{alg:lq}}{$\ell=\ell+1$ and add $a^\ell$ to $Q$ \label{alg:lq1}\;}
                 }
        }
    }
    \Return $M$
\end{algorithm}

\begin{procedure}
\caption{DecideAccRej($a^\ell,q_a,b,q_b$)}\label{proce:Decide}
    \DontPrintSemicolon
    \SetAlgoLined
    \If{$a^x\in M(b)$ for $x<\ell$ \label{dec:art1}}{$M=(M\setminus \{(a^x,b)\}) \cup \{(a^\ell,b)\}$\label{dec:art2}}
    \uElseIf{$|M(b)|<q_b$\label{dec:st1}}{$M=M\cup \{(a^\ell,b)\}$\;}
    \uElseIf{$|M(b)|==q_b$}{
        Let $a_j^y\in M(b)$ be the least preferred vertex in $M(b)$\;
        \If{($\ell>y$) or ($\ell==y$ and $a$ precedes $a_j$ in $\prefb$)}{
        $M=M\setminus \{(a_j^y,b)\}\cup \{(a^\ell,b)\}$ and add $a_j^y$ to $Q$ if $\forall x \ a_j^x\notin Q$\;}\label{dec:st2}}
        \If{ $|M(a)|<q_a $ and  $\forall x\  a^x\notin Q$\label{dec:que1}}  {add $a^\ell$ to $Q$ \label{dec:que2}
        }  
\end{procedure}

We note that $s$ and $t$ are both $O(|E|)$ and each edge of $G$ is explored at most $\S+\T+2$ times. Thus, the running time of our algorithm is $O(|E|^2)$.

We use the example shown in Figure~\ref{fig:exCritical} to illustrate the execution of Algorithm~\ref{algo:maxPopSC2LQ}. Note that $\S=4$ and $\T=1$. Thus, during the course of algorithm an \emph{lq}-vertex from side $\AA$ can raise its level up to 6 (i.e. $\S+\T+1$) whereas a non \emph{lq}-vertex from side $\AA$ can raise its level up to 2 (i.e. $\T+1$). A possible proposal sequence for Algorithm~\ref{algo:maxPopSC2LQ} is shown in Table~\ref{tab:propSeq} in the appendix. Initially, each vertex $a\in\AA$ is at level 0 and proposes to vertices in $\preflqa$. These correspond to Proposals  1, 2 and 3 in Table~\ref{tab:propSeq}. Note that $a_1,a_2$ and $a_3$ do not propose to $b_1$ at level 0 since $b_1$ is a non \emph{lq}-vertex. Once the vertex $a$ raises its level to 1, it proposes to all the vertices in $\prefa$.  The capacity of $b_2$ changes from 1 to 2 when none of its lower quota many matched partners is at level 0 (see proposal number 7 in Table~\ref{tab:propSeq}). The vertex $a_2$ raises its level to 3 since it remains deficient even after it has exhausted its preference list at level 2. Subsequently, $a_1$ and $a_3$ also raise their respective levels to 3.  Note that the capacity of $a_1$ changes from 2 to 1 when it raises its level to 3 and above (see proposal number 16). 

Throughout Algorithm~\ref{algo:maxPopSC2LQ}, at most, one level copy of a vertex $a$ can be matched to $b$. Thus, we say $(a, b) \in M$ instead of $(a^x, b) \in M$ whenever the context is clear.
The next lemma states important properties of matching~$M$. 
\begin{lemma}\label{lem:merged}
Let $a\in\AA$, $b\in\BB$, and $(a,b)\in E\setminus M$. Then the following properties hold:
\begin{enumerate}
    \item\label{lem:defStudQminus1}
If $|M(a)|>q^-(a)$, then the level achieved by $a$ during the course of Algorithm~\ref{algo:maxPopSC2LQ} is at most $\T+1$.

\item\label{lem:defbQm1}
Let $\Tilde{a}\in\AA$ be such that 
$(\Tilde{a}^x,b)\in M$ for some $x<\T$, then $|M(b)|\le q^-(b)$.  

\item\label{lem:undersA}
If $|M(a)|<q^+(a)$, then $|M(b)|=q^+(b)$ and for each $\Tilde{a}^x\in M(b)$ we have $x\ge \T+1$.
\item\label{lem:dfStud}
If $|M(a)|<q^-(a)$, then $|M(b)|=q^+(b)$ and for each $\Tilde{a}^x\in M(b)$ we have $x=\S+\T+1$.

\item\label{lem:noedg}
If vertex $a$ reaches \emph{a} level $x>1$, then for each $\Tilde{a}^y\in M(b)$ we have $y\ge x-1$.
\end{enumerate}
\end{lemma}

\begin{proof}
\begin{itemize}
    \item \textbf{Proof of~\ref{lem:defStudQminus1}:} Note that a vertex $a\in\AA$ is allowed to increase its level above $\T+1$ only if $a$ remains deficient after $a^{\T+1}$ exhausts $\prefa$ (Line~\ref{alg:lq} and~\ref{alg:lq1} of Algorithm~\ref{algo:maxPopSC2LQ}). Once $a$ is at level above $\T+1$ the capacity of $a$, $c(a)$, is set to $q^-(a)$ (Line~\ref{alg:LQstudents}). Thus, if any $a\in\AA$ is at level above $\T+1$, then $|M(a)|\le q^-(a)$. This implies if $|M(a)|>q^-(a)$ then the level achieved by $a$ is at most $\T+1$.
    
    \item \textbf{Proof of~\ref{lem:defbQm1}:} For the sake of contradiction, let us assume that $(\Tilde{a}^x,b)\in M$ for some $x<\T$ but $|M(b)|> q^-(b)$. This implies  that the procedure \texttt{DecideAccRej} is invoked with capacity of $b$ equal to $q^+(b)$ where $q^+(b)>q^-(b)$. Note that the procedure \texttt{DecideAccRej} can be invoked with capacity of $b$ equal to $q^+(b)$ at Line~\ref{alg:22} or at Line~\ref{alg:24} in the Algorithm~\ref{algo:maxPopSC2LQ}. Let us consider the first invocation of  \texttt{DecideAccRej} where the capacity of $b$ is set to $q^+(b)$. We remark that the first invocation cannot happen at Line~\ref{alg:24} because Line~\ref{alg:24} requires $|M(b)|>q^-(b)$. So, let us assume that the first invocation is at Line~\ref{alg:22}. Note that at Line~\ref{alg:22} every vertex in $M(b)$ is at level at least $\T$, and after which during the course of the algorithm, $b$ accepts proposals only from vertices in $\AA$ which are at level at least $\T$. This contradicts the fact that $\Tilde{a}^x$ for $x<\T$ is in $M(b)$ at the end of the algorithm.

    \item \textbf{Proof of~\ref{lem:undersA}:} Let us assume that $a$ remains under-subscribed and there exists a neighbour $b$ not matched to $a$. The fact that $a$ remains under-subscribed implies that $a^{\T+1}$ has exhausted $\prefa$. That is, $a^{\T+1}$ proposed to $b$. Thus, if $b$ is under-subscribed, then $b$ must have accepted the proposal from $a^{\T+1}$. It is given that $(a,b)\notin M$ thus $|M(b)|=q^+(b)$. Furthermore, no vertex $\Tilde{a}$ in $M(b)$ can be at level $x<\T+1$ because otherwise $b$ would have rejected one of the vertices in $M(b)$ to accept the proposal of $a^{\T+1}$.  
    
    \item \textbf{Proof of~\ref{lem:dfStud}:} This proof is very similar to the proof of~\ref{lem:undersA} above.   The fact that $|M(a)|<q^-(a)$ implies that $a^{\S+\T+1}$ has exhausted $\prefa$ during the course of Algorithm~\ref{algo:maxPopSC2LQ}. If $b$ is under-subscribed, then it must have accepted the proposal from $a^{\S+\T+1}$. This implies that $(a,b)\in M$, a contradiction.  Thus, $|M(b)|=q^+(b)$. Now, let us assume that there exists $\Tilde{a}\in \AA$ such that $(\Tilde{a}^x,b)\in M$ for some $x<\S+\T+1$. But then $b$ would have rejected $\Tilde{a}^x$ to accept the proposal of $a^{\S+\T+1}$ as it prefers a higher level vertex to a vertex at lower level. This implies that $(a,b)\in M$, a contradiction.
    
    \item \textbf{Proof of~\ref{lem:noedg}:} Recall that a vertex $b\in\BB$ always prefers $a$ over $\Tilde{a}$ if $a$ is at higher level than that of $\Tilde{a}$. Let us assume for the sake of contradiction that there exists $\Tilde{a}\in \AA$ such that $(\Tilde{a}^y,b)\in M$ for $y<x-1$. The fact that $(a,b)\in E$ and $a$ achieves the level $x$ implies that $a$ failed to fulfil its capacity after $a^{x-1}$ exhausted its preference list $\prefa$ or $\preflqa$ as appropriate.  Note that if $b$ receives a proposal from a vertex $\Tilde{a}\in\AA$ at levels below $x-1$ then $b$ is also available to receive proposals from vertices in $\AA$ at levels $\ge x-1$. This is because when a vertex  in $\AA$ transitions to higher level, it proposes to possibly a superset of vertices that it proposes to in the lower level (recall that $\prefa$ is a superset of $\preflqa$). Furthermore, if a vertex in $\BB$ receives a proposal from some vertex in $\AA$ at a level $z$ then
    it receives a proposal from all vertices proposing in $\AA$ at level $z$.
    Since $b$ is matched to a vertex at level $y<x-1$, it must be the case that $b$ has received a proposal from $a^{x-1}$ and it accepted this proposal. Also, $b$ cannot reject $a^{x-1}$ till $(\Tilde{a}^y,b)\in M$ for $y<x-1$. We know that $\Tilde{a}^y\in M(b)$.  Thus, $(a,b)\in M$ and we get a contradiction to the fact that $(a,b)\notin M$. 
\end{itemize}
This completes the proofs of all parts of Lemma~\ref{lem:merged}.
\end{proof}

Lemma~\ref{lem:merged}(\ref{lem:undersA}) implies that 
for any edge $(a,b)\in E\setminus M$ either $a$ or $b$ is fully-subscribed. Thus we have following claim.

\begin{cl}\label{cl:fully-subscribed}
The matching output by Algorithm~\ref{algo:maxPopSC2LQ} is a maximal matching.
\end{cl}

%% file: 3cloneG.tex
\section{Cloned graph and criticality of $M$} 
\label{sec:criticality}
In this section, our goal is to prove that the matching $M$ computed by Algorithm~\ref{algo:maxPopSC2LQ} is critical. 
In order to establish criticality as well as popularity (in the next section), we construct a {\em cloned} graph $G_M$ inspired by the similar construction in~\cite{brandl2019two}.

\subsection{Construction of the cloned graph $G_M$}\label{sec:cloneG}
The cloned graph is constructed using the matching $M$ computed by Algorithm~\ref{algo:maxPopSC2LQ} and allows us to work with a one-to-one matching $M^*$ corresponding to $M$. We illustrate the construction using an example. For the instance $G$ in Figure~\ref{subfig:exCrit},  Algorithm~\ref{algo:maxPopSC2LQ} outputs the matching $M=\{(a_1,b_1),(a_2,b_2),(a_3,b_2)\}$    (see Table~\ref{tab:propSeq} for an execution sequence). Figure~\ref{fig:Clonedgraph} shows the cloned graph $G_M$.

\vspace{0.1in}

\noindent{\textbf{Vertex set of $G_M:$ }} The vertex set of the graph $G_M$ consists of $(\AA' \cup \BB' \cup \tL\cup \dL)$. 
\begin{enumerate}
    \item The sets $\AA'\cup\BB'$ contain $q^+(v)$ many \emph{clones} corresponding to each vertex $v\in \AA\cup\BB$. That is, 
\begin{eqnarray*}
\AA' = \{ \ a_j :  a \in \AA \ \mbox{and } 1 \le j \le q^+(a)\} \hspace{0.5in}
\BB' = \{ \ b_j :  b \in \BB \  \mbox{and } 1 \le j \le q^+(b)\}\\
\end{eqnarray*}
Corresponding to the vertex $a_1$ in Figure~\ref{subfig:exCrit}, the cloned graph $G_M$ in Figure~\ref{fig:Clonedgraph} contains $q^+(a_1)=2$ clones, namely, $a_{11}$ and $a_{12}$. Similarly,  for the vertex $b_1$ it contains one clone $b_{11}$.

\item The set $ \tL = \tLA \cup \tLB$ denotes the set of last-resorts. For each vertex $v \in \AA\cup \BB$ we introduce $f(v)=q^+(v) - q^-(v)$ many last-resorts. These last-resorts are specific to $v$ and will be connected to only a subset of clones of $v$. That is, 
\begin{eqnarray*}
\tLA = \{ \ \ell_{a_j}  \ :    a \in \AA \ \mbox { and } 1 \le j \le f(a)\} \hspace{0.5in}
\tLB = \{ \ \ell_{b_j}  \ :   b \in  \BB \  \mbox{and } 1 \le j \le f(b)\}
\end{eqnarray*}
Corresponding to the vertex $b_2$ in Figure~\ref{subfig:exCrit},  the cloned graph $G_M$ in Figure~\ref{fig:Clonedgraph} contains $q^+(b_2)-q^-(b_2)=1$ last-resort, namely $\ell_{b_2}$. Note that a vertex with its lower quota equal to the upper quota does not have any last-resorts corresponding to it -- see for example $a_2$ and $a_3$. 

\item The set $\dL=\dLA\cup\dLB$ denotes the set of dummy vertices where  $|\dLA|=\Dfa{M}$ and $|\dLB|=\Dfb{M}$. Note that $\dLA$ (or $\dLB$) is non-empty if and only if there exists a deficient vertex $v\in\AA$ (or $v\in \BB$ respectively) in the matching $M$. For example, $G_M$ in Figure~\ref{fig:Clonedgraph} has $\dLA=\{d_a\}$ and $\dLB=\emptyset$.
\end{enumerate}

Note that the bipartition of the graph $G_M$ is $(\AA'\cup\tLB\cup\dLB)\cup (\BB'\cup\tLA\cup\dLA)$.
The clones, last-resorts and dummy vertices allow us to convert the many-to-many matching $M$ to an $(\AA' \cup \BB'\cup\dL)$-perfect one-to-one matching~$M^*$. Now we define the edge set of the graph $G_M$.

\begin{figure}
\centering
\begin{subfigure}{.5\textwidth}
  \centering
  \scalebox{0.75}{\input{Gm1}}
  \caption{}
  \label{fig:Clonedgraph}
\end{subfigure}%
\begin{subfigure}{.5\textwidth}
  \centering
  \scalebox{0.75}{\input{Gm2}}
  \caption{ }
  \label{subfig:critPerf}
\end{subfigure}
\caption{ (i) The cloned graph $G_{M}$ corresponding to the critical matching $M=\{(a_1,b_1),(a_2,b_2),(a_3,b_2)\}$ for the graph shown in Figure~\ref{fig:exCritical}. Note that $f(v)=0$ for $v\in\{a_2,a_3\}$, $f(v)=1$ for $v\in\{a_1,b_1,b_2\}$, $\Dfa{M}=1$ and $\Dfb{M}=0$. Also, note that $|M(a_2)|<q^-(a_2)$, $|M(v)|=q^-(v)=1$ for $v\in\{a_1,a_3\}$, and $|M(v)|>q^-(v)$ for $v\in\{b_1,b_2\}$. Bold blue edges denote the $(\AA'\cup\BB'\cup \dL)$-perfect matching $M^*$. (ii) Bold red edges represent the $(\AA'\cup\BB')$-perfect matching $N^*$ for a critical matching $N=\{(a_1,b_2),(a_2,b_1),(a_2,b_2)\}$.
}
\end{figure}
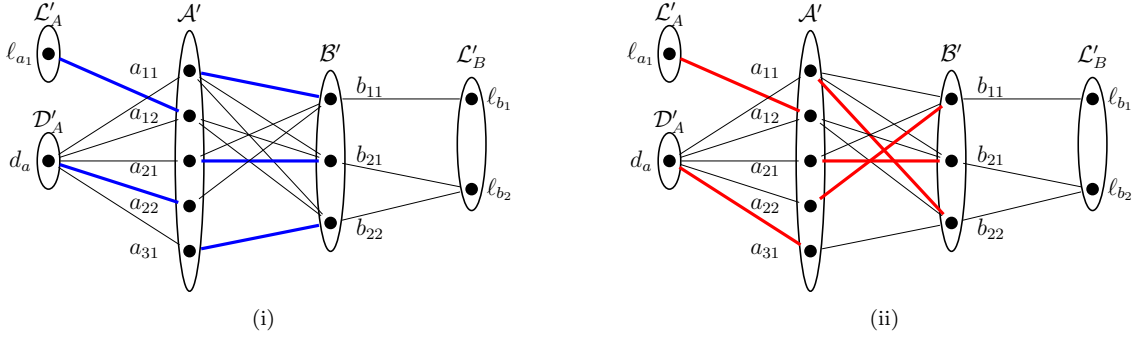

\vspace{0.1in}

\noindent {\bf The edge set $E'$:} The edge set $E'= M^* \cup E'_U$ where $M^*$ denote the set of matched edges, $E'_U$ denote the set of unmatched edges and are defined as follows --
\begin{enumerate}
    \item {\textbf{The matching $M^*$:}} Given the many-to-many matching $M$, we construct a one-to-one matching $M^*$ in $G_M$ as follows. 
    \begin{itemize}
        \item For every edge $(a, b) \in M$ we select an unselected clone of $a$, say $a_i$, and an unselected clone of $b$, say $b_j$, and add the edge $(a_i, b_j)$ to $M^*$. For example, corresponding to the edge $(a_2,b_2)\in M$ in Figure~\ref{subfig:exCrit},  $M^*$ in Figure~\ref{fig:Clonedgraph} contains $(a_{21},b_{21})$.
        \item If a vertex $v\in \AA\cup\BB$ is deficient in $M$ then let $v_j$ be one of its $q^-(v)-|M(v)|$ many unmatched clones.  We select a unique unmatched dummy vertex $d_k\in\dL$ and add the edge $(v_j,d_k)$ to $M^*$. Such an unmatched dummy vertex must exist because we have $\Dfa{M}=|\dLA|$ and $\Dfb{M}=|\dLB|$. For example, since $a_2$ remains deficient in $M$, the matching $M^*$ in Figure~\ref{fig:Clonedgraph} contains  the edge $(a_{22},d_{a})$.
        \item For all other ($q^+(v)-\max\{q^-(v),|M(v)|\}$) unmatched clones of a vertex $v\in\AA\cup\BB$ in $M$ let $v_j$ be one of these unmatched clones. We select a unique last-resort $\ell_{v_j}$ and add the edge $(v_j, \ell_{v_j})$ to $M^*$. For example, since $a_1$ remains under-subscribed in $M$, the matching  $M^*$ in Figure~\ref{fig:Clonedgraph} contains $(a_{12},\ell_{a_1})$ .
    \end{itemize}  
    Note that the number of dummy vertices is exactly equal to $\Df{M}$ and the number of last-resorts corresponding to a vertex $v$ is exactly equal to $f(v)=q^+(v)-q^-(v)$. We remark that each dummy vertex in $\dL$ must get matched to some clone $v_i\in\AA'\cup\BB'$ such that the corresponding vertex $v$ is deficient in $M$. This is because we first match the clones of each vertex up to its deficiency to dummy vertices, and the total number of dummy vertices is exactly equal to $\Df{M}$. Thus, $M^*$ is an $(\AA'\cup\BB'\cup \dL)$-perfect matching.

    \item {\textbf{The unmatched edges $E'_U$ in $G_M$:} \label{itm:unmatched}}  For every edge $(a, b) \in E \setminus M$, we have all edges corresponding to the complete bipartite graph between the clones of $a$ and clones of $b$. For example -- $(a_1,b_2)\notin M$. Thus, $G_M$ contains a complete bipartite graph between $\{a_{11},a_{12}\}$ and $\{b_{11},b_{12}\}$ (see Figure~\ref{fig:Clonedgraph}). We create a complete bipartite graph between all clones in $\AA'$ and all dummy vertices in $\dLA$. Similarly, we create a complete bipartite graph between all clones in $\BB'$ and all dummy vertices in $\dLB$. 
    
    We also have unmatched edges from clones of vertices to the last-resorts corresponding to that vertex. Here, we consider two cases  for a vertex $v\in\AA\cup\BB$ depending on whether
$|M(v)| > q^-(v)$ or $|M(v)| \le q^-(v)$. This construction is important for our dual feasible setting in Section~\ref{sec:dfeas}.
\begin{itemize} 
\item For a vertex $v$ where $|M(v)|  > q^-(v)$ we have a complete bipartite graph between all the $q^+(v)$ many clones of $v$ and all the $f(v)$ many last-resorts corresponding to $v$. Thus, we add to $E'_U$ edges of the form $(\ell_{v_k}, v_j)$ where $1 \le k \le f(v)$ and $1 \le j \le q^+(v)$. For example, $|M(b_2)|=2>q^-(b_2)=1$ and hence, $G_M$ contains complete bipartite graph between $\{b_{11},b_{12}\}$ and the corresponding last-resort $\ell_{b_2}$.

\item  For a vertex $v$ where $|M(v)| \le q^-(v)$, we have a complete bipartite graph between the set of clones of $v$ matched to last-resorts and all the $f(v)$ many last-resorts corresponding to $v$. For instance, $|M(a_1)|=1\le q^-(a_1)$ and the clone $a_{11}$ is not matched to last-resort. Thus, $a_{11}$ is \emph{not} connected to $\ell_{a_1}$.

Thus, we add to $E'_U$ edges of the form
$(\ell_{v_k}, v_j)$ where $1 \le k \le f(v)$ and $v_j$ is matched to a last-resort in the above construction. Note that if a vertex is deficient, then some of its clones are matched to dummy vertices, and such clones are not connected to last-resorts.
\end{itemize}
\end{enumerate}



We are yet to prove that $M$ is a critical matching and hence, for any critical matching $N$, we claim $\Df{N}\le \Df{M}=|\dL|$. Now, in Lemma~\ref{lem:crit2perf} we show that the graph $G_M$ allows us to map \emph{any} critical matching $N$ in $G$ to an $(\AA'\cup\BB')$-perfect one-to-one matching $N^*$ in $G_M$. We emphasize that, at this stage, we are not claiming that $N$ can be mapped to an $(\AA'\cup\BB'\cup \dL)$-perfect matching.  Once we prove that $M$ is critical (in Lemma~\ref{lem:CriticalM}) then it will be clear that every critical matching $N$ in $G$ can be indeed mapped to an $(\AA'\cup\BB'\cup \dL)$-perfect one-to-one matching $N^*$ in $G_M$ (as shown in Lemma~\ref{lem:corr}). Consider the critical matching $N=\{(a_1,b_2),(a_2,b_1),(a_2,b_2)\}$ shown in Figure~\ref{fig:exCritical}. A possible mapping of $N$ in $G_M$ is $N^*= \{(a_{11},b_{22}),(a_{12},\ell_{a_1}),(a_{21},b_{21}),(a_{22},b_{11}),(a_{31},d_a)\}$ (see Figure~\ref{subfig:critPerf}).

Before we prove Lemma~\ref{lem:crit2perf}, it will be useful to prove an important property (Claim~\ref{lem:noLessDef}) of a critical matching. We remark that this is a property of a critical matching and is independent of our algorithm.

\begin{cl}\label{lem:noLessDef}
Let $N_1$ be any critical matching and $N_2$ be any matching in $G$. Then $\Dfa{N_1}\le \Dfa{N_2}$ and $\Dfb{N_1}\le \Dfb{N_2}$. 
\end{cl}
\begin{proof}
Here we will prove that $\Dfa{N_1}\le\Dfa{N_2}$. The proof for showing  $\Dfb{N_1}\le\Dfb{N_2}$ is symmetric. For the sake of contradiction, assume that $\Dfa{N_2}<\Dfa{N_1}$.  Consider the symmetric difference $N_1\oplus N_2$. Since we are dealing with many-to-many matchings, $N_1\oplus N_2$ is a collection of connected components. It is possible to decompose each component of $N_1\oplus N_2$ into maximal $N_1-N_2$ alternating paths and cycles (Section~2.2~~\cite{nasre2017popularity}).  Since $\Dfa{N_2}<\Dfa{N_1}$, there must exists an alternating path $\rho=\langle u,u'\ldots,v\rangle$ with $u\in\AA$ such that using $\rho$ we obtain a matching $N_1'=N_1\oplus \rho$ such that $\Dfa{N_1'}<\Dfa{N_1}$.  Note that $(u,u')\in N_2$, and $u$ has a property that $|N_1(u)|< q^-(u)$. Since $\rho$ is a maximal $N_1-N_2$ alternating path and it ends with an $N_2$ edge hence $|N_1(u)|<|N_2(u)|$. We consider the two cases below depending on the length of the path $\rho$.

If $\rho$ is of odd length then it is an augmenting path with respect to $N_1$. That is, the deficiency of $u$ in $N_1'$ is reduced by one, whereas the deficiencies of all other vertices, except $v$, remain same as in $N_1$. Note that the deficiency of $v$ does not increase in $N_1'$ because $\rho$ is an augmenting path for $N_1$. Thus, $N_1'$ satisfies the property that $\Df{N_1'}<\Df{N_1}$. This contradicts the criticality of~$N_1$. 

If $\rho$ is of even length then due to the maximality of $\rho$, the other endpoint $v\in\AA$ is such that $|N_2(v)|< |N_1(v)|$ and $\rho$ ends with an $N_1$-edge. Now we consider three different cases --
\begin{itemize}
    \item[(i)] $|N_2(v)|< q^-(v)$ and $|N_1(v)|\le q^-(v)$: In this case, the deficiency of vertex $v$ in $N_1'$ is increased by one and the deficiency of vertex $u$ in $N_1'$ is reduced by one, whereas the deficiencies of all other vertices remain same. Note that both the endpoints $u$ and $v$ of $\rho$ are in $\AA$. Thus, $\Dfa{N_1'}=\Dfa{N_1}$ which contradicts the fact that $\Dfa{N_1'}<\Dfa{N_1}$.
    \item[(ii)] $|N_2(v)|< q^-(v)$ and $|N_1(v)|> q^-(v)$: In this case, the deficiency of $v$ in $N_1'$ remains same as in $N_1$ because $v$ is surplus in $N_1$. Note that $|N_1(u)|< q^-(u)$ which implies that the deficiency of $u$ in $N_1'$ is reduced by one. The deficiencies of all other vertices remain same. Hence, $\Df{N_1'}<\Df{N_1}$ which contradicts the criticality of $N_1$.
    \item[(iii)] $|N_2(v)|\ge q^-(v)$: In this case, $v$ is surplus in $N_1$ because $|N_1(v)|> |N_2(v)|$. Thus, the same argument as in (ii) above works.
 \end{itemize}
 Thus, we conclude that such a path $\rho$ does not exist and hence $\Dfa{N_1}\le\Dfa{N_2}$.
 \end{proof}

\begin{lemma}\label{lem:crit2perf}
For every critical matching $N$ in $G$ there exists an $(\AA'\cup\BB')$-perfect one-to-one matching $N^*$ in $G_M$.
\end{lemma}

\begin{proof}
For any critical matching $N$ in $G$, we construct an $(\AA' \cup \BB')$-perfect one-to-one matching $N^*$ in $G_M$ as follows.
\begin{itemize}
          \item[(i)] \textbf{Matching edges common to both $M$ and $N$:} For each edge $(a,b)\in M \cap N$, we add the edge $(a_i,b_j)$ to $N^*$ such that $(a_i,b_j)\in M^*$.
           \item[(ii)] \textbf{Matching edges in $N$ but not in $M$:} For every edge $(a,b)\in N\setminus M$, we find an unmatched clone $a_i$ of $a$ and $b_j$ of $b$, and add $(a_i,b_j)$ to $N^*$ as given below. Note that for an edge $(a,b)\notin M$, $G_M$ has a complete bipartite graph between all the clones of $a$ and all the clones of $b$. Also, there must exist some unmatched clone of $a$ and unmatched clone of $b$ because we have $q^+(a)$ and $q^+(b)$ many clones of $a$ and $b$ respectively. 
          \begin{itemize}
              \item Amongst all the unmatched clones of $a$ we pick a clone $a_i$ which is not adjacent to last-resorts, if such a clone exists. If all the unmatched clones of $a$ are adjacent to last-resorts, then we choose an  arbitrary unmatched clone~$a_{i}$.
              \item Amongst all the unmatched clones of $b$ we pick a clone $b_j$ which is not adjacent to last-resorts, if such a clone exists. If all the unmatched clones are adjacent to last-resorts, then we choose an arbitrary unmatched clone~$b_{j}$.
          \end{itemize}
          \item[(iii)] \textbf{Matching unmatched clones to dummy vertices:} Consider a vertex $v\in\AA\cup\BB$ such that $v$ is deficient in $N$. By the construction of the cloned graph $G_M$, if $v$ is not surplus in $M$  then $q^-(v)$ many clones are not connected to any last-resort. If $v$ is surplus in $M$ then \emph{all} clones of $v$ are connected to all the last-resorts corresponding to $v$. Thus, after steps (i) and (ii) above, we select all the unmatched clones of $v$ which are not adjacent to any last-resorts and match them to arbitrary but distinct dummy vertices in $\dL$. We add these corresponding edges to $N^*$.
          
          \item[(iv)] \textbf{Matching unmatched clones to last-resorts:} For any vertex $v_k\in \AA' \cup \BB'$ that is left unmatched after above steps, we select an arbitrary but distinct unmatched last-resort, say $\ell_{v_j}$, and add the edge $(v_k,\ell_{v_j})$ in~$N^*$.
        \end{itemize}
        
 It remains to show that $N^*$ is ($\AA'\cup\BB'$)-perfect. To do so, we will show that for every $v\in\AA\cup\BB$, all the clones of $v$ are matched in $N^*$. Recall $N$ is a critical matching in $G$ and thus by Claim~\ref{lem:noLessDef} $\Dfa{N}\le\Dfa{M}$ and $\Dfb{N}\le \Dfb{M}$. The fact that $|\dLA|=\Dfa{M}$ and $|\dLB|=\Dfb{M}$ implies that $\Dfa{N}\le |\dLA|$ and $\Dfb{N}\le |\dLB|$. By construction of $G_M$, all the dummy vertices in $\dLA$ are adjacent to all the clones in $\AA'$, and all the dummy vertices in $\dLB$ are adjacent to all the clones in $\BB'$.

Note that in step (ii) and step (iii) all the clones which are not adjacent to last-resorts are given priority to get matched to true vertices or to dummy vertices while constructing $N^*$. Further, we have enough dummy vertices ($|\dLA|\ge \Dfa{N}$ and $|\dLB|\ge\Dfb{N}$) as stated above. Thus, for a vertex $v$, at most $q^-(v)$ many clones which are not adjacent to last-resorts must get matched in $N^*$ either to some true clones (in $\AA'\cup\BB'$) or to some dummy vertices. For the vertex $v$ graph $G_M$ contains exactly $q^+(v)-q^-(v)$ many last-resorts and hence, all the copies of $v$ which remain unmatched after step (iii) must get matched to some of these last-resorts in step (iv). So, all the clones of $v$ are matched in $N^*$.   \end{proof}

Now we proceed to prove that $M$ is a critical matching. In order to do so it will be useful to partition the vertices of $G_M$ into subsets and establish properties about these partitions.

\begin{definition}[True Edge]\label{def:trueEd}
An edge $(a_i,b_j)$ in $G_M$ is called a true edge if none of the endpoints $a_i$ and $b_j$ is in $\tL\cup\dL$.
\end{definition}

\noindent \textbf{Partition of vertices:} 
Now we partition the vertex set $\AA'\cup\BB'\cup\tL\cup\dL$ as described next. The vertex set $\AA'\cup\tLB\cup\dLB$ is partitioned into $\AA'_0\cup\AA'_1\cup\ldots\cup\AA'_{\S+\T+1}$. Similarly, the vertex set $\BB'\cup\tLA\cup\dLA$ is partitioned into $\BB'_0\cup\BB'_1\cup\ldots\cup\BB'_{\S+\T+1}$. Since $M^*$ is an $(\AA'\cup\BB'\cup\dL)$-perfect matching, each clone in $\AA'$ is matched. A clone $a_i\in \AA'$ can be matched to a clone in $\BB'$, or a last-resort, or a dummy vertex. If a clone $a_i$ of the vertex $a\in\AA$ is matched to $b_j\in\BB'$ in $M^*$ then there exists an edge $(a^x,b)\in M$. 
The clone $a_i$ is assigned the partition $\AA'_x$ based on the level $x$ corresponding to the matched edge $(a^x,b)$.
\begin{itemize}
    \item \label{itm:part1} \textbf{Clones in $\AA'$ which are matched along true edges:} Let $a\in \AA, b\in \BB$ and $(a^x,b)\in M$ for some $0\le x\le \S+\T+1$ and let $(a_i,b_j)$ denote the corresponding edge in $M^*$. We add $a_i$ to $\AA'_x$ and $b_j$ to $\BB'_x$.

    \item \label{itm:part2}{\bf Clones in $\AA'$ which are matched to dummy vertices:} Consider $a\in \AA$ which is deficient in $M$. Such a vertex has $q^-(a)- |M(a)|$ many clones matched to dummy vertices in $M^*$. 
    We add all the $q^-(a)-|M(a)|$ many clones of each $a\in\AA$ that are matched to dummy vertices to $\AA'_{\S+\T+1}$ and their matched dummy vertices to $\BB'_{\S+\T+1}$. 
    
    \item \label{itm:part3}{\bf Clones in $\BB'$ which are matched to dummy vertices:} Consider $b\in \BB$ which is deficient in $M$. Such a vertex has $q^-(b)- |M(b)|$ many clones matched to dummy vertices in $M^*$. 
    We add all the $q^-(b)-|M(b)|$ many clones of each $b\in\BB'$ that are matched to dummy vertices to $\BB'_0$, and their matched dummy vertices to $\AA'_0$. 
    
    \item \label{itm:part4}{\bf Clones in $\AA'$ which are matched to last-resorts:} If $a\in \AA$ is under-subscribed but not deficient in $M$ then $q^+(a)- |M(a)|$ many clones of $a$ are matched to last-resorts in $M^*$. If $a\in \AA$ is deficient in $M$ then $q^+(a)-q^-(a)$ many clones are matched to last-resorts. This implies that for each $a\in\AA$, exactly $q^+(a)-\max\{q^-(a), |M(a)|\}$ many clones are matched to last-resorts. We add all these clones of $a$ that are matched to last-resorts to $\AA'_{\T+1}$ and their matched last-resorts to $\BB'_{\T+1}$.
    
    \item \label{itm:part5}\textbf{Clones in $\BB'$ which are matched to last-resorts:} Consider $b\in\BB$. As described above we have exactly $q^+(b)-\max\{q^-(b), |M(b)|\}$ many clones matched to last-resorts. We add all these clones of $b$ to $\BB'_{\T}$, and their matched last-resorts to $\AA'_{\T}$.
    
    \item \label{itm:part6} \textbf{Unmatched last-resorts:} We add all the remaining last-resorts $\ell_{a}\in \tLA$ to the partition set $\BB'_{\T+1}$, and all the remaining last-resorts $\ell_{b}\in \tLB$ to the partition set $\AA'_{\T}$.
\end{itemize}

It is convenient to visualize the partitions from top to bottom in decreasing order of the indices --  see Figure~\ref{fig:maxlevelgraph}. 
The edges of $M^*$ are horizontal in the figure.  
We state the properties of the clones and edges in $G_M$ with respect to the partition.

\begin{lemma}\label{pr:vertexCloneG} Let $a\in\AA$, $b\in\BB$ and  $u$ be any clone of $a$  and $v$ be any clone of $b$  then the following properties hold.
\begin{enumerate}
    \item \label{itm:partprop2} The partition $\bigcup_{x=\T+2}^{\S+\T+1}\AA'_x$ contains at most $q^-(a)$ many clones of $a$. (follows from the contrapositive statement of Lemma~\ref{lem:merged}(\ref{lem:defStudQminus1}))
    \item \label{itm:partprop3} The partition $\bigcup_{x=0}^{\T-1}\BB'_x$ contains at most $q^-(b)$ many clones of $b$. (follows from Lemma~\ref{lem:merged}(\ref{lem:defbQm1}))
    \item \label{itm:partprop4} The total number of clones of vertices in $\AA$ at levels $\T+2$ and higher are upper bounded by $\S$, that is,  $|(\bigcup_{x=\T+2}^{\S+\T+1}\AA'_x)\cap \AA'|\le \S$ (follows from~\ref{itm:partprop2} above). The total number of clones of vertices in $\BB$ at levels $\T-1$ and lower are upper bounded by $\T$, that is, $|(\bigcup_{x=0}^{\T-1}\BB'_x)\cap\BB'|\le \T$. (follows from~\ref{itm:partprop3} above)
    \item \label{itm:partprop7} If  $u\in \AA'_{\S+\T+1}$ and $M^*(u)\in \dLA$ then all the neighbours of $u$ in $G_M$ are present only in the partition $\BB'_{\S+\T+1}$. (follows from Lemma~\ref{lem:merged}(\ref{lem:dfStud}))
    \item \label{itm:partprop8} If $M^*(u)\in \tLA$ then all the neighbours of $u$ in $G_M$ are present in a partition $\BB'_{x}$ where $x\ge\T+1$. (follows from Lemma~\ref{lem:merged}(\ref{lem:undersA}))
\end{enumerate}
\end{lemma}
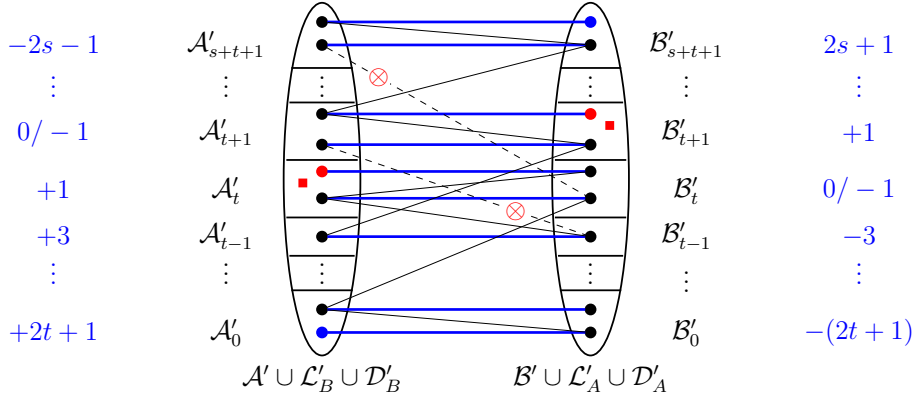
\begin{figure}
\begin{center}
    \scalebox{0.85}{\input{maxLevelGraph}}
\end{center}

\caption{The graph $G_M$ corresponding to $M$. Black circles represent clones in $\AA'\cup\BB'$, Blue circles represent dummy vertices in $\dLA\cup\dLB$ and red circles and squares represent last-resorts in $\tLA\cup\tLB$. Matched vertices are represented by circles and unmatched vertices are represented by squares. The blue horizontal lines represent matched edges in $M^*$. Solid black lines represent edges which are not matched in $M^*$. Dashed black lines marked with crossed red circles represent steep downward edges that are not present in $G_M$ (Lemma~\ref{pr:noedge}).
Integers in blue represent dual values and are relevant in Section~\ref{sec:pop}.
}
\label{fig:maxlevelgraph}
\end{figure}

\begin{lemma}\label{lem:noedgeLresC}
Let $a\in \AA$ have a clone $a_j$ in $G_M$.  If $a_j\in \AA'_x$ for $x>\T+1$ then  $a_j$ is not adjacent to any last-resorts in $G_M$.
\end{lemma}
\begin{proof}
Since $a_j\in\AA_x'$ for $x>\T+1$ the corresponding vertex $a$ must have achieved level at least $\T+2$. This implies that $|M(a)|\le q^-(a)$ by Lemma~\ref{lem:merged}(\ref{lem:defStudQminus1}). Since $|M(a)|\le q^-(a)$, our construction ensures that only those clones  which are matched to last-resorts are connected to last-resorts. The clones which are matched to last-resorts are added to $\AA'_{\T+1}$ while partitioning. Thus, $a_j$ is not adjacent to any last-resorts in~$G_M$.
  \end{proof}

\begin{lemma}\label{lem:noedgeLresS}
Let $b\in\BB$ have a clone $b_j$ in $G_M$. If $b_j\in\BB'_x$ for $x<\T$ then $b_j$ is not adjacent to any last-resorts.
\end{lemma}
\begin{proof}
Assume that $b_j$ is a clone of some vertex $b\in\BB$ and $b_j\in\BB'_x$ for $x<\T$. Recall that $f(b)$ denotes the difference between $q^+(b)$ and $q^-(b)$. If $f(b)=0$ then last-resorts corresponding to $b$ do not exist. Hence, let us assume that $f(b)>0$. First we show that $|M(b)|\le q^-(b)$. For the sake of contradiction let us assume that $|M(b)|> q^-(b)$.  Using the contrapositive of Lemma~\ref{lem:merged}(\ref{lem:defbQm1}) we see that none of the clones of any vertex in $M(b)$ are in $\AA'_x$ for $x< \T$. This contradicts the fact that $b_j\in \BB'_x$ for $x<\T$. Thus, we establish that $|M(b)|\le q^-(b)$. 

The construction of $G_M$ ensures that only those clones which are matched to last-resorts are connected to last-resorts. While partitioning, we added all such clones to $\BB'_{\T}$. Thus, all the clones $b_j\in\BB'_x$ for $x<\T$ are not adjacent to last-resorts.
  \end{proof}

Let $(a, b)\in E'$ be  an edge such that $a \in \AA'_x$ and $b\in \BB'_y$. We say that such an edge is of the form $\AA'_x \times \BB'_y$. Also, assume that the edges in $G_M$ are implicitly directed from $\AA'$ to $\BB'$. Lemma~\ref{pr:noedge} below gives one important property about the edges which cannot be present in $G_M$. An edge of the from $\AA'_x \times  \BB'_y$ such that $x>y+1$ is referred to as a \emph{steep downward} edges as it goes at least two levels down.
\begin{lemma}\label{pr:noedge}
The graph $G_M$ does not contain steep downward edges. That is, there is no edge in $G_M$ of the form
$\AA'_x \times  \BB'_y$ such that $x>y+1$.
\end{lemma}
\begin{proof}
Here, we show that no edge in $G_M$ is of the form $\AA'_x \times  \BB'_y$ such that $x>y+1$. We have three different kinds of edges in $G_M$ -- true edges, edges incident to last-resorts and edges incident to dummy vertices. By the construction of $G_M$, all the dummy vertices in $\dLA$ are added to $\BB'_{\S+\T+1}$ and all the dummy vertices in $\dLB$ are added to $\AA'_{0}$. That is, all the dummy vertices are either in bottom-left partition or in top-right partition. This implies that no edge incident to a dummy vertex can be of the form $\AA'_x \times  \BB'_y$ such that $x>y+1$. By the construction of $G_M$, all the last-resorts in $\tLA$ are added to $\BB'_{\T+1}$ and all the last-resorts in $\tLB$ are added to $\AA'_{\T}$. Thus, Lemma~\ref{lem:noedgeLresC} and Lemma~\ref{lem:noedgeLresS} ensure that no edge incident to a last-resort is of the form $\AA'_x \times  \BB'_y$ such that $x>y+1$. Lemma~\ref{lem:merged}(\ref{lem:noedg}) states that if a vertex $a$ is at level $x>1$ and $b$ is its neighbour such that $(a,b)\notin M$ then no vertex in $M(b)$ is at level below $x-1$. This implies that all the clones of $b$ are in $\BB'_z$ for $z\ge x-1$. If a true edge $(a,b)\in M$ then by construction, $(a,b)\in\AA'_x \times  \BB'_x$ for some $0\le x\le \S+\T+1$. Thus, we conclude that the graph $G_M$ does not contain steep downward edges.
\end{proof}

%% file: Gm1.tex
\begin{tikzpicture}[thick,
  lnode/.style={draw=white,fill=white},
  fsnode/.style={draw,circle,fill=black,scale=0.5},
  every fit/.style={ellipse,draw,inner sep=-2pt,text width=0.5cm},
  -,shorten >= 3pt,shorten <= 3pt
]

  \node[fsnode] (a11) at (0,3) {};
  \node[fsnode] (a12) at (0,2.2) {};
  \node[fsnode] (a21) at (0,1.4) {};
  \node[fsnode] (a22) at (0,0.6) {};
  \node[fsnode] (a31) at (0,-0.2) {};
  \node[fsnode] (b11) at (2.5,2.5) {};
  \node[fsnode] (b21) at (2.5,1.4) {};
  \node[fsnode] (b22) at (2.5,0.3) {};
  
  \node at (-0.8,3) {$a_{11}$};
\node at (-0.8,2.2) {$a_{12}$};
\node at (-0.8,1.25) {$a_{21}$};
\node at (-0.8,0.6) {$a_{22}$};
\node at (-0.8,-0.2) {$a_{31}$};
\node at (3.2,2.72) {$b_{11}$};
\node at (3.2,1.5) {$b_{21}$};
\node at (3.2,0.2) {$b_{22}$};

  \node[fsnode] (la1) at (-2.5,3.3) {};
  \node[fsnode] (lb1) at (5,2.5) {};
  \node[fsnode] (lb2) at (5,0.9) {};
  \node at (-3,3.3) {$\ell_{a_{1}}$};
  \node at (5.5,2.5) {$\ell_{b_{1}}$};
  \node at (5.5,0.9) {$\ell_{b_{2}}$};
  
   \node[fsnode] (da1) at (-2.5,1.4) {};
    \node at (-3,1.4) {$d_{a}$};

\node [fit=(a11) (a31),label=above:$\AA'$] {};
\node [fit=(b11) (b22),label=above:$\BB'$] {};
\draw (-2.5, 3.3) ellipse (0.2cm and 0.5cm);
\node at (-2.5, 4) {$\tLA$};

\draw (-2.5, 1.4) ellipse (0.2cm and 0.5cm);
\node at (-2.5, 2.1) {$\dLA$};

\node [fit=(lb1) (lb2),label=above:$\tLB$] {};

\draw[ultra thick, blue](a11) -- (b11);
\draw[ultra thin](a11) -- (b21);
\draw[ultra thin](a11) -- (b22);
\draw[ultra thin](a12) -- (b21);
\draw[ultra thin](a12) -- (b22);
\draw[ultra thin](a21) -- (b11);
\draw[ultra thick, blue](a21) -- (b21);
\draw[ultra thin](a22) -- (b11);
\draw[ultra thick, blue](a31) -- (b22);

\draw[ultra thick, blue](la1) -- (a12);
\draw[ultra thin](b11) -- (lb1);
\draw[ultra thin](b21) -- (lb2);
\draw[ultra thin](b22) -- (lb2);
\draw[ultra thin](da1) -- (a11);
\draw[ultra thin](da1) -- (a12);
\draw[ultra thin](da1) -- (a21);
\draw[ultra thick,blue](da1) -- (a22);
\draw[ultra thin](da1) -- (a31);

\end{tikzpicture}

%% file: Gm2.tex
\begin{tikzpicture}[thick,
  lnode/.style={draw=white,fill=white},
  fsnode/.style={draw,circle,fill=black,scale=0.5},
  every fit/.style={ellipse,draw,inner sep=-2pt,text width=0.5cm},
  -,shorten >= 3pt,shorten <= 3pt
]

  \node[fsnode] (a11) at (0,3) {};
  \node[fsnode] (a12) at (0,2.2) {};
  \node[fsnode] (a21) at (0,1.4) {};
  \node[fsnode] (a22) at (0,0.6) {};
  \node[fsnode] (a31) at (0,-0.2) {};
  \node[fsnode] (b11) at (2.5,2.5) {};
  \node[fsnode] (b21) at (2.5,1.4) {};
  \node[fsnode] (b22) at (2.5,0.3) {};
  
  \node at (-0.8,3) {$a_{11}$};
\node at (-0.8,2.2) {$a_{12}$};
\node at (-0.8,1.25) {$a_{21}$};
\node at (-0.8,0.6) {$a_{22}$};
\node at (-0.8,-0.2) {$a_{31}$};
\node at (3.2,2.72) {$b_{11}$};
\node at (3.2,1.5) {$b_{21}$};
\node at (3.2,0.2) {$b_{22}$};

  \node[fsnode] (la1) at (-2.5,3.3) {};
  \node[fsnode] (lb1) at (5,2.5) {};
  \node[fsnode] (lb2) at (5,0.9) {};
  \node at (-3,3.3) {$\ell_{a_{1}}$};
  \node at (5.5,2.5) {$\ell_{b_{1}}$};
  \node at (5.5,0.9) {$\ell_{b_{2}}$};
  
   \node[fsnode] (da1) at (-2.5,1.4) {};
    \node at (-3,1.4) {$d_{a}$};

\node [fit=(a11) (a31),label=above:$\AA'$] {};
\node [fit=(b11) (b22),label=above:$\BB'$] {};
\draw (-2.5, 3.3) ellipse (0.2cm and 0.5cm);
\node at (-2.5, 4) {$\tLA$};

\draw (-2.5, 1.4) ellipse (0.2cm and 0.5cm);
\node at (-2.5, 2.1) {$\dLA$};

\node [fit=(lb1) (lb2),label=above:$\tLB$] {};
\draw[ultra thin](a11) -- (b11);
\draw[ultra thin](a11) -- (b21);
\draw[ultra thick, red](a11) -- (b22);
\draw[ultra thin](a12) -- (b21);
\draw[ultra thin](a12) -- (b22);
\draw[ultra thin](a21) -- (b11);
\draw[ultra thick, red](a21) -- (b21);
\draw[ultra thick, red](a22) -- (b11);
\draw[ultra thin](a31) -- (b22);

\draw[ultra thick, red](la1) -- (a12);
\draw[ultra thin](b11) -- (lb1);
\draw[ultra thin](b21) -- (lb2);
\draw[ultra thin](b22) -- (lb2);
\draw[ultra thin](da1) -- (a11);
\draw[ultra thin](da1) -- (a12);
\draw[ultra thin](da1) -- (a21);
\draw[ultra thin](da1) -- (a22);
\draw[ultra thick, red](da1) -- (a31);

\end{tikzpicture}

%% file: maxLevelGraph.tex
	\begin{tikzpicture}[scale=0.6, thick,fsnode/.style={draw,circle,fill=black,scale=0.4}, lnode/.style={draw=red,circle,fill=red,scale=0.4}, snode/.style={draw=blue,circle,fill=blue,scale=0.4},tnode/.style={draw=red,fill=red,scale=0.4}]
	\draw (1,3) ellipse (1cm and 4.6cm);
	\node at (1,-2.2) {$\AA'\cup\tLB\cup\dLB$};
	\draw (0.2,5.9) -- (1.8,5.9);
	\node at (1,5.6) {\vdots};
	\draw (0.15,5) -- (1.85,5);
	\draw (0.07,3.5) -- (1.93,3.5);
	\draw (0.07,2) -- (1.93,2);
	\draw (0.15,1) -- (1.85,1);
	\node at (1,0.7) {\vdots};
	\draw (0.2,0.1) -- (1.8,0.1);
	\node at (-1.5,6.5) {$\AA'_{\S+\T+1}$};
	\node at (-1.5,5.6) {\vdots};
	\node at (-1.5,4.2) {$\AA'_{\T+1}$};
	\node at (-1.5,2.7) {$\AA'_{\T}$};
	\node at (-1.5,1.5) {$\AA'_{\T-1}$};
	\node at (-1.5,0.7) {\vdots};
	\node at (-1.5,-1) {$\AA'_{0}$};
	\node at (-6,6.5) {\clr{$-2\S-1$}};
	\node at (-6,5.6) {\clr{\vdots}};
	\node at (-6,4.2) {\clr{$0/-1$}};
	\node at (-6,2.7) {\clr{$+1$}};
	\node at (-6,1.5) {\clr{$+3$}};
	\node at (-6,0.7) {\clr{\vdots}};
	\node at (-6,-1) {\clr{$+2\T+1$}};
	\draw (8,3) ellipse (1cm and 4.6cm);
	\node at (8,-2.2) {$\BB'\cup\tLA\cup\dLA$};
	\draw (7.2,5.9) -- (8.8,5.9);
	\node at (8,5.6) {\vdots};
	\draw (7.15,5) -- (8.85,5);
	\draw (7.07,3.5) -- (8.93,3.5);
	\draw (7.07,2) -- (8.93,2);
	\draw (7.15,1) -- (8.85,1);
	\node at (8,0.7) {\vdots};
	\draw (7.2,0.1) -- (8.8,0.1);
	\node at (10.5,6.5) {$\BB'_{\S+\T+1}$};
	\node at (10.5,5.6) {\vdots};
	\node at (10.5,4.2) {$\BB'_{\T+1}$};
	\node at (10.5,2.7) {$\BB'_{\T}$};
	\node at (10.5,1.5) {$\BB'_{\T-1}$};
	\node at (10.5,0.5) {\vdots};
	\node at (10.5,-1) {$\BB'_0$};
	\node at (15,6.5) {\clr{$2\S+1$}};
	\node at (15,5.6) {\clr{\vdots}};
	\node at (15,4.2) {\clr{$+1$}};
	\node at (15,2.7) {\clr{$0/-1$}};
	\node at (15,1.5) {\clr{$-3$}};
	\node at (15,0.7) {\clr{\vdots}};
	\node at (15,-1) {\clr{$-(2\T+1)$}};
    \node[fsnode] (a1) at (1,7.1) {};
	\node[fsnode] (a2) at (1,6.5) {};
	\node[fsnode] (b1) at (8,6.5) {};
	\draw[very thick,blue] (a2) -- (b1);
	\node[fsnode] (b2) at (8,3.9) {};
	\node[fsnode] (a3) at (1,3.9) {};
	\draw[very thick,blue] (a3) -- (b2);
	\node[fsnode] (b3) at (8,2.5) {};
	\node[fsnode] (a4) at (1,2.5) {};
	\draw[very thick,blue] (a4) -- (b3);
	\node[fsnode] (a5) at (1,1.5) {};
	\node[fsnode] (b4) at (8,1.5) {};
	\draw[very thick,blue] (a5) -- (b4);
	\node[fsnode] (a6) at (1,-0.4) {};
	\node[fsnode] (b5) at (8,-0.4) {};
	\draw[very thick,blue] (a6) -- (b5);
	\node[fsnode] (b6) at (8,-1) {};
	\node[tnode] (lb1) at (8.5,4.4) {};
	\node[tnode] (la1) at (0.5,2.9) {};
	\node[lnode] (lb2) at (8,4.7) {};
	\node[lnode] (la2) at (1,3.2) {};
	
	\node[snode] (db1) at (8,7.1) {};
	\node[snode] (da1) at (1,-1) {};
	\node[fsnode] (a7) at (1,4.7) {};
	\draw[very thick,blue] (lb2) -- (a7);
	\node[fsnode] (b7) at (8,3.2) {};
	\draw[very thick,blue] (la2) -- (b7);
	\draw[very thick,blue] (db1) -- (a1);
	\draw[very thick,blue] (da1) -- (b6);
	\draw[ultra thin] (a1) -- (b1);
	\draw[ultra thin] (a4) -- (b4);
	\draw[ultra thin] (a5) -- (b2);
	\draw[ultra thin] (a6) -- (b3);
	\draw[ultra thin] (a6) -- (b6);

	\draw[ultra thin] (a7) -- (b2);
	\draw[ultra thin] (a7) -- (b1);
	\draw[ultra thin] (b7) -- (a4);

\draw[ultra thin, dashed] (a2) -- (b3)node[pos=0.2,fill=white,inner sep=0.1pt]{\textcolor{red}{$\otimes$}};
    
\draw[ultra thin, dashed] (a3) -- (b4)node[pos=0.73,fill=white,inner sep=0.1pt]{\textcolor{red}{$\otimes$}};

	\end{tikzpicture}

%% file: 4crit.tex
 \subsection{Criticality of matching $M$}\label{subsec:critical}
In this section, we show that the matching $M$ output by Algorithm~\ref{algo:maxPopSC2LQ} is critical. Recall that $(a,b)\in M$ if and only if $(a_i,b_j)\in M^*$ for some clone $a_i$ and $b_j$. It will be useful to work with the one-to-one matching in which clones assigned to either last-resorts or dummy vertices in $M^*$ are considered unmatched. Therefore, we consider the matching $M'$ such that $M'=M^*\setminus\{(v_k,\ell_{v})\ :\ v_k\in\AA'\cup\BB' \mbox{ and  }\ell_v\in\tL\cup\dL\}$

\begin{lemma}\label{lem:CriticalM}
The output matching $M$ of Algorithm~\ref{algo:maxPopSC2LQ} is critical for $G$.
\end{lemma}

\begin{proof} 
We prove the criticality of $M$ using the graph $G_M$. Note that $M^*$ is an $(\AA'\cup\BB'\cup\dL)$-perfect matching in $G_M$. Let $N$ be any critical matching in $G$ and, by Lemma~\ref{lem:crit2perf}, let $N^*$ be the corresponding ($\AA'\cup\BB'$)-perfect one-to-one matching in $G_M$. We obtain a one-to-one matching $M'$ from $M^*$ and $N'$ from $N^*$ by removing all the last-resorts and dummy vertices. Thus, $M'=M^*\setminus\{(v_k,\ell_{v})\ :\ v_k\in\AA'\cup\BB' \mbox{ and  }\ell_v\in\tL\cup\dL\}$ and  $N'=N^*\setminus\{(v_k,\ell_{v})\ :\ v_k\in\AA'\cup\BB' \mbox{ and  }\ell_v\in\tL\cup\dL\}$.  Here we show that there is no alternating path $\rho$ in $G_M$ with respect to $M'$ such that $M'\oplus \rho$ has a lesser deficiency than $M'$. The proof is divided into two parts called the $\AA$-part and the $\BB$-part, where we show the criticality of $M$ for the respective parts. 
 
\noindent \textbf{Proof of ($\AA$-part):} For the sake of contradiction, let us assume that $M$ is not critical for the $\AA$-part, that is, $\Dfa{N}<\Dfa{M}$. This implies that there exists an alternating path $\rho$ in $M'\oplus N'$ such that $N'$ matches more clones representing lower-quota positions on $\rho$ than $M'$. Let $\rho=\langle u_0, v_1,u_1,v_2,u_2,\ldots,v_k,u_k,\ldots,u_p\rangle$, where $(v_i,u_i)\in M'$ and the other edges of $\rho$ are in the matching $N'$. The idea is to show that the path begins at $u_0$ which is at the highest level $\S+\T+1$ and has at least the next clone on $\AA'$-side, that is $u_1$, also at the same level  (see Figure~\ref{fig:critLem}). Finally, the path ends at level $\T+1$ or below. We further show that the number of clones along this path in levels $\T+2, \ldots \S+\T+1$ is more than the total number of clones that can be accommodated in these levels. This gives us the desired contradiction.

Let the first vertex $u_0$ represents a clone of some vertex $a\in\AA$ such that deficiency of $a$ in $N$ is less than the deficiency of $a$ in $M$ and $M'(u_0)=\bot$ (that is, $M^*(u_0)\in \tLA\cup\dLA$). Now we show that we can always assume that $\rho$ starts at the highest level, that is, $u_0\in\AA'_{\S+\T+1}$.  To show this let us consider $M^*(u_0)$. We have two possibilities -- either $M^*(u_0)\in\dLA$ or $M^*(u_0)\in\tLA$. If $M^*(u_0)\in\dLA$ then $u_0\in\AA'_{\S+\T+1}$ (by construction) and the path $\rho$ starts at the highest level as desired. If $M^*(u_0)\in\tLA$ then $u_0\in\AA'_{\T+1}$.  In this case, since $a$ is deficient in $M$ there must exists some clone of $a$, say $u'$, such that $u'\in \AA'_{\S+\T+1}$ and $M^*(u')\in\dLA$ ($a^{\S+\T+1}$ must have exhausted $\prefa$).  We claim that $u'$ is adjacent to $v_1$ because of the following reason. Let us assume that $v_1$ represents a clone of $b\in\BB$. The fact that $(a,b)\notin M$ implies that each clone of $a$, in particular $u'$, is adjacent to each clone of $b$, in particular $v_1$. Thus, we can replace $u_0$ by $u'$ and $\rho$ starts at the highest level $\S+\T+1$ as desired. Note that this change of vertex does not affect the deficiency of the resultant matching obtained by switching along $\rho$. Since $M^*(u_0)\in\dLA$, Lemma~\ref{pr:vertexCloneG}(\ref{itm:partprop7}) implies that all the neighbours of $u_0$ are only in $\BB'_{\S+\T+1}$. Thus, we conclude that $v_1\in \BB'_{\S+\T+1}$, and since $(u_1,v_1)\in M'$, $u_1$ must be in $\AA'_{\S+\T+1}$ (see Figure~\ref{fig:critLem}). The other end of $\rho$ can be a clone in $\AA'$ or a clone in $\BB'$. We consider both these cases below.

\begin{figure}
\begin{center}
    \scalebox{0.85}{\input{lemCrit}}
\end{center}

\caption{Blue colored edges denote the edges in $M'$ whereas red edges denote the edges in $N'$.}
\label{fig:critLem}
\end{figure}
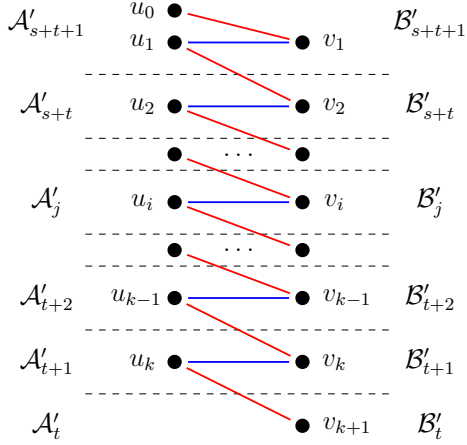

\noindent\textbf{The path $\rho$ ends at some clone in $\AA'$:} Suppose that the path ends at a clone in $\AA'_x$ for $x>\T+1$. By the construction of the graph $G_M$, all the clones in $\AA'_x$ for $x>\T+1$ represent lower-quota positions for some \emph{lq}-vertex. Thus, if $\rho$ ends at a clone $u_i$ such that $u_i\in \AA'_{x}$ for $x>\T+1$ then $N\oplus \rho$ matches the same number of lower-quota positions. This contradicts that $N'$ matches strictly more number of lower-quota positions along $\rho$. This implies that the other endpoint  of $\rho$ must be in $\AA'_{x}$ for $x\le\T+1$.
By Lemma~\ref{pr:noedge} we know that $G_M$ does not contain steep downward edges. Hence, we conclude that $v_2\in \BB'_x$ and $M'(v_2)=$ $u_2\in\AA'_x$ for $x\in\{\S+\T+1,\S+\T\}$. The absence of steep downward edges implies that 
if $u_i\in\AA'_y$ and $u_{i+1}\in\AA'_x$ then $x-y\le 1$ for all indices $i$ on $\rho$ and hence, $\rho$ must contain at least one clone from each $\AA'_x$ for $\T+1\le x\le \S+\T+1$ (see Figure~\ref{fig:critLem}). Let us assume that $u_k\in\AA'_{\T+1}$. Thus, $u_{k-1}\in \AA'_y$ for $y\le \T+2$. We remark that the last vertex $u_p$ and the starting vertex $u_0$ of $\rho$, both cannot be the clones of the vertex $a$,  otherwise, we get a contradiction that $N'$ matches strictly more number of lower-quota positions along $\rho$. Now notice that there are at least two clones $u_0$ and $u_1$ in $\AA'_{\S+\T+1}$ and at least one clone in each $\AA'_x$ for $\T+2\le x\le \S+\T$. Thus, there are at least $\S+1$ many clones in $\bigcup_{x=\T+2}^{\S+\T+1}\AA'_x$ in $G_M$ contradicting Lemma~\ref{pr:vertexCloneG}(\ref{itm:partprop4}).

\noindent\textbf{The path $\rho$ ends at some clone $v_{k+1}\in\BB'$:} This implies that $M'(v_{k+1})=\bot$ and hence $v_{k+1}\in \BB'_y$ for $y\le \T$. Applying Lemma~\ref{pr:noedge} we see that $u_k\in\AA'_x$ for $x\le \T+1$. Now repeating the same argument as in the previous case we find that there are at least $\S+1$ many clones in $\bigcup_{x=\T+2}^{\S+\T+1}\AA'_x$ in $G_M$. Thus, we get a contradiction. This completes the proof that $M$ is critical for vertices in $\AA$. 

\vspace{0.15in}

\noindent \textbf{Proof of ($\BB$-part):}
Proof of this part is similar to that of $\AA$-part. For the sake of contradiction, let us assume that $\Dfb{N}<\Dfb{M}$. This implies that there exists an alternating path $\rho$ in $M'\oplus N'$ such that $N'$ matches more clones representing lower-quota positions on $\rho$ than $M'$. Let $\rho=\langle v_0, u_1, v_1, u_2, v_2, \ldots, u_k, v_k,\ldots,v_p \rangle$ where $(u_i,v_i)\in M'$ and the other edges of $\rho$ are in the matching $N'$. Furthermore, assume that the first vertex $v_0$ represents a clone of some vertex $b\in\BB$ such that the deficiency of $b$ in $N$ is less than the deficiency of $b$ in $M$ and $M'(v_0)=\bot$. As in the proof of $\AA$-part, we can assume that $\rho$ starts at the lowest level, that is, $v_0\in\BB'_0$. 
Since $b$ is deficient, it must not have received enough ($\ge q^-(b)$ many) proposals from its neighbours. In other words, there does not exist a vertex $a$ such that $(a,b)\in E\setminus M$ and $a$ is at a level above 0. This implies $u_1\in\AA'_0$ and $v_1=M'(u_1)$ is in $\BB'_0$. The other end of $\rho$ can be in $\BB'$ or in $\AA'$. We consider both these cases below.

\noindent\textbf{The path $\rho$ ends at some clone in $\BB'$:} Suppose that the path ends at a clone in $\BB'_x$ for $x<\T$. By the construction of $G_M$, all the clones in $\BB'_x$ for $x<\T$ represent lower-quota positions for some \emph{lq}-vertex.  Thus, if $\rho$ ends at a clone $v_i$ such that $v_i\in \BB'_{x}$ for $x<\T$ then $N\oplus \rho$ matches the same number of lower-quota positions. This contradicts the fact that $N'$ matches strictly more number of lower-quota positions along $\rho$. This implies that the other endpoint of $\rho$ cannot be in $\BB'_x$ for $x<\T$. 
Lemma~\ref{pr:noedge} implies that if $v_i\in\BB'_x$ and $v_{i+1}\in\BB'_y$ then $y-x\le 1$ for all indices $i$ on $\rho$ and hence, $\rho$ must contain at least one clone from each $\BB'_x$ for $1\le x\le \T-1$. Let us assume that $v_k\in\BB'_{\T}$. Thus, $v_{k-1}\in\BB'_y$ for $y\ge \T-1$.  We remark that the last vertex $v_p$ and the starting vertex $v_0$ of $\rho$, both cannot be clones of the vertex $b$; otherwise, we get a contradiction that $N'$ matches more lower-quota positions along $\rho$ than $M'$. We observe that $\rho$ contains at least two clones $v_0$ and $v_1$ from $\BB'_0$, and at least one clone from each $\BB'_x$ for $1\le x\le \T-1$.  Thus, the path  $\rho$ contains at least $\T+1$ many clones from $\bigcup_{x=0}^{\T-1}\BB'_x$ which contradicts the total number of clones accommodated in these levels by Lemma~\ref{pr:vertexCloneG}(\ref{itm:partprop4}).  

\noindent\textbf{The path $\rho$ ends at some clone in $\AA'$:} Note that $M'(u_{k+1})=\bot$ and by construction unmatched clones of a vertex $a\in\AA$ are only in $\AA'_x$ for $x\ge \T+1$. Thus, $v_k\in \BB'_y$ for $y\ge \T$. Using the same argument as in the previous case, we show that $\rho$ contains at least $\T+1$ many clones from $\bigcup_{x=0}^{\T-1}\BB'_x$ to get a contradiction. Hence, we conclude that such a path $\rho$ cannot exist.   \end{proof}

%% file: lemCrit.tex
\begin{tikzpicture}[thick,
  lnode/.style={draw=white,fill=blue},
  fsnode/.style={draw,circle,fill=black,scale=0.5},
  every fit/.style={ellipse,draw,inner sep=-2pt,text width=0.5cm},
  -,shorten >= 3pt,shorten <= 3pt
]


  \node[fsnode] (u0) at (0,0.5) {};
  \node at (-0.5,0.5) {$u_0$};
  \node[fsnode] (v1) at (2,0) {};
  \node at (2.5,0) {$v_1$};
  \node[fsnode] (u1) at (0,0) {};
  \node at (-0.5,0) {$u_1$};
  \node at (-2,0.3) {$\AA'_{\S+\T+1}$};
  \node at (4,0.3) {$\BB'_{\S+\T+1}$};
   \draw[ultra thin, dashed](-1.5,-0.5) -- (3.5,-0.5);
   
  \node[fsnode] (v2) at (2,-1) {};
  \node at (2.5,-1) {$v_2$};
  \node[fsnode] (u2) at (0,-1) {};
  \node at (-0.5,-1) {$u_2$};
   \node at (-2,-1) {$\AA'_{\S+\T}$};
   \node at (4,-1) {$\BB'_{\S+\T}$};
   \draw[ultra thin,dashed](-1.5,-1.5) -- (3.5,-1.5);
   
   \node[fsnode] (ui1) at (0,-1.75) {};
   \node[fsnode] (vi1) at (2,-1.75) {};
   \node at (1,-1.75) {$\ldots$};
   \draw[ultra thin,dashed](-1.5,-2) -- (3.5,-2);
   
  \node[fsnode] (vi) at (2,-2.5) {};
  \node at (2.5,-2.5) {$v_i$};
  \node[fsnode] (ui) at (0,-2.5) {};
  \node at (-0.5,-2.5) {$u_i$};
   \node at (-2,-2.5) {$\AA'_{j}$};
  \node at (4,-2.5) {$\BB'_{j}$};
   \draw[ultra thin, dashed](-1.5,-3) -- (3.5,-3);
   
   \node[fsnode] (ui11) at (0,-3.25) {};
   \node[fsnode] (vi11) at (2,-3.25) {};
   \node at (1,-3.25) {$\ldots$};
   \draw[ultra thin, dashed](-1.5,-3.5) -- (3.5,-3.5);

\node[fsnode] (vk1) at (2,-4) {};
  \node at (2.7,-4) {$v_{k-1}$};
  \node[fsnode] (uk1) at (0,-4) {};
  \node at (-0.6,-4) {$u_{k-1}$};
   \node at (-2,-4) {$\AA'_{\T+2}$};
  \node at (4,-4) {$\BB'_{\T+2}$};
   \draw[ultra thin, dashed](-1.5,-4.5) -- (3.5,-4.5);
   
   \node[fsnode] (vk) at (2,-5) {};
  \node at (2.5,-5) {$v_{k}$};
  \node[fsnode] (uk) at (0,-5) {};
  \node at (-0.5,-5) {$u_{k}$};
   \node at (-2,-5) {$\AA'_{\T+1}$};
  \node at (4,-5) {$\BB'_{\T+1}$};
  \draw[ultra thin, dashed](-1.5,-5.5) -- (3.5,-5.5);
  
   \node[fsnode] (vk11) at (2,-6) {};
  \node at (2.7,-6) {$v_{k+1}$};
   \node at (-2,-6) {$\AA'_{\T}$};
  \node at (4,-6) {$\BB'_{\T}$};

 \draw[thick, red ](u0) -- (v1);
 \draw[thick, blue](u1) -- (v1);
 \draw[thick, red ](u1) -- (v2);
 \draw[thick, blue](u2) -- (v2);
 \draw[thick, blue](ui) -- (vi);
 \draw[thick, red ](u2) -- (vi1);
 \draw[thick, red ](ui1) -- (vi);
 \draw[thick, red ](ui) -- (vi11);
 \draw[thick, red ](ui11) -- (vk1);
 \draw[thick, blue](uk1) -- (vk1);
 \draw[thick, red ](uk1) -- (vk);
 \draw[thick, blue](uk) -- (vk);
\draw[thick, red ](uk) -- (vk11);

\end{tikzpicture}

%% file: 5popular.tex
\section{Popularity of $M$}
\label{sec:pop}
In this section, we prove that the matching $M$ output by Algorithm~\ref{algo:maxPopSC2LQ} is a maximum size popular matching in the set of critical matchings. To capture the votes of the vertices,  we assign a weight to every edge of the cloned graph $G_M$. The weight of an edge $(a_i,b_j)$ in $G_M$, denoted by $wt(a_i,b_j)$, captures the votes of $a_i$ and $b_j$ when these vertices cast their votes for the edge $(a_i,b_j)$ compared to their respective partners in $M^*$. 
The weights of edges in $G_M$ are as follows:
\begin{enumerate}
    \item For a true edge $(a_i,b_j)$  where $a_i\in\AA'$ and $b_j \in \BB'$:
    \begin{eqnarray*}
 		wt(a_i,b_j)=\begin{cases}
			0   \qquad\qquad\qquad\qquad\qquad\qquad\qquad\qquad\qquad  \mbox{ if $(a_i,b_j)\in M^*$} \\
			vote_a(b_j, M^*(a_i)) + vote_b(a_i, M^*(b_j))\ \ \qquad \mbox{ otherwise}
		\end{cases}
	\end{eqnarray*}

    \item  For an edge $(u,v)$ where $u\in\AA'\cup\BB'$ and $v\in\tL\cup \dL$:
    \begin{eqnarray*}
 		wt(u,v)=\begin{cases}
		\ \ 0\   \qquad \mbox{ if $M^*(u)\in \tL\cup\dL$} \\
			-1\ \qquad \mbox{ otherwise}
		\end{cases}
	\end{eqnarray*}
    
\end{enumerate}
The weight of an edge of $G_M$ denotes the
sum of the votes of endpoints when compared to the matching $M^*$. Thus, for $ e \in E'$, we have $-2 \le wt(e) \le 2$.

\begin{lemma}\label{lem:edgewts}
     Let $e=(a_i,b_j)$ be any edge in $G_M$ where $a_i\in \AA'$ and $b_j\in \BB'$. Then,
      \begin{enumerate}
         \item \label{itm:y1} If $e\in \AA'_x\times \BB'_{x-1}$ then $wt(e)=-2$.
         \item \label{itm:y2} If $e\in \AA'_x\times \BB'_x$ then $wt(e) \le 0$.
     \item \label{itm:y3} If $e\in \AA'_x\times \BB'_y$ for $y>x$ then $wt(e)\leq 2$. 
      \end{enumerate}
 \end{lemma}
\begin{proof}
	\begin{itemize}
	      \item \textbf{Proof of \ref{itm:y1}:} If $e\in \AA'_x\times \BB'_{x-1}$ then $(a,b)\notin  M$ because all the matched edges are of the form $\AA'_y\times\BB'_y$.
	      Let us assume that $M^*(b_j)=\Tilde{a}_k$ where $\Tilde{a}_k$ is a clone of $\Tilde{a}\in\AA$. Note that $b_j\in\BB'_{x-1}$ implies $(\Tilde{a}^{x-1},b)\in M$.  Recall that during our proposal-based algorithm, a vertex $a\in\AA$ proposes to a subset of vertices in $\prefa$ at certain lower levels. Now we show that $a^{x-1}$ must have proposed to $b$ during our algorithm irrespective of whether $b$ is an \emph{lq}-vertex or a non \emph{lq}-vertex. Note that $a_i\in\AA_x'$ implies $a^{x-1}$ has exhausted its preference list. Suppose that $b$ is an \emph{lq}-vertex. Thus, $b$ appears in the preference list of all its neighbours at all the levels. This implies that $a^{x-1}$ must have proposed to $b$. Now suppose that $b$ is a non \emph{lq}-vertex. Since $\Tilde{a}^{x-1}\in M(b)$ it must be the case that $\Tilde{a}^{x-1}$ proposes to $b$. This implies  $x-1\ge \T$ which further implies that $a^{x-1}$ must have proposed to $b$. 
	      
	      Since $a\notin M(b)$, the vertex $b$ must have rejected $a^{x-1}$. We claim that the votes of both the endpoints $a$ and $b$ must be $-1$. Suppose for contradiction that $vote_{a}(b_j,M^*(a_i))=1$. Here, we consider two cases -- (i) $a^x$ does not propose to $b$ which implies that $vote_{a}(b_j,M^*(a_i))=-1$ and we get a contradiction. (ii) $a^x$ proposes to $b$. In this case, $b$ prefers a higher level vertex $a$ over a lower level vertex $\Tilde{a}$ and thus $(a,b)\in M$ which is a contradiction to the fact that $(a,b)\notin  M$.  Thus, $vote_{a}(b_j,M^*(a_i))\not=1$. Now assume that $vote_{b}(a_i,M^*(b_j))=1$ then $b$ must have accepted the proposal of $a^{x-1}$ or $a^x$ because $\Tilde{a}^{x-1}\in M(b)$ and $b$ prefers $a$ over $\Tilde{a}$.  This implies $(a,b)\in M$ which is a contradiction. Thus, we conclude that $vote_{a}(b_j,M^*(a_i))=vote_{b}(a_i,M^*(b_j))=-1$. Hence, $wt(e)=-2$. 
	      
            
            \item \textbf{Proof of \ref{itm:y2}}: If $e\in \AA'_x \times \BB'_x$ then we consider two cases-- (i) if $e\in M^*$ then by construction $wt(e)=0$. (ii) If $e\notin M^*$ then we claim that $wt(e)$ cannot be $+2$. Assume for contradiction that $wt(e)=wt(a_i,b_j)=2$. Also assume that $M^*(a_i)=b'_k$ where $b'_k$ is a clone of $b'\in\BB$ and $M^*(b_j)=a_r'$ where $a_r'$ is a clone of $a'$. The clone $b_j$ at level $x$ and $M^*(b_j)=a_r'$ implies that $b$ received a proposal from $a'$ at level $x$, therefore, $a$ can also propose to $b$ at level $x$. The assumption that $wt(a_i,b_j)=2$ implies that $a$ prefers $b$ over $b'$. Thus, $a^x$ proposed to $b$ before proposing to $b'$ and $b$ accepts it. 
            This is a contradiction. Note that for a true edge $e$,  $wt(e)\in \{-2,0,2\}$ and as shown above $wt(e)\not=+2$. Thus, we conclude that $wt(e)\le 0$.
            
	       \item \textbf{Proof of \ref{itm:y3}:} The maximum possible weight for any edge $e$ is 2. 
		\end{itemize}
		
		This completes the proof of the lemma.\end{proof}

Recall that in Lemma~\ref{lem:crit2perf}, we showed that a critical matching $N$ in $G$ can be mapped to an  ($\AA'\cup\BB'$)-perfect one-to-one matching $N^*$ in $G_M$. Having established that the output $M$ of Algorithm~\ref{algo:maxPopSC2LQ} is a critical matching in $G$, we now show that a critical matching $N$ in $G$ can be mapped to an  ($\AA'\cup\BB'\cup\dL$)-perfect one-to-one matching $N^*$ in $G_M$ such that $wt(N^*)$ is exactly equal to the difference in the votes received by matchings $N$ and $M$ according to the corresponding functions $\cor$.

\begin{lemma}\label{lem:corr}
For every critical matching $N$ in $G$ and a correspondence function $\cor$ there exists an ($\AA'\cup\BB'\cup\dL$)-perfect matching $N^*$ in $G_M$ such that $wt(N^*)=\Delta(N,M,\cor)$.
\end{lemma}

\begin{proof}
      For the given critical matching $N$ in $G$, we construct an $(\AA' \cup \BB'\cup\dL)$-perfect matching $N^*$ in $G_M$ and show that $wt(N^*)=\Delta(N, M,\cor)$, where $wt(N^*)$ denotes the sum of the weights of the edges in $N^*$. We construct $N^*$ using the correspondence function $\cor$. We find appropriate clones $a_i$ of $a$ and $b_j$ of $b$ corresponding to each edge $(a,b)\in N$, where $(a_i,b_j)\in E' \cap N^*$. For a deficient $a$ and a deficient $b$ in $N$, at least $q^-(a)-|N(a)|$ and at least $q^-(b)-|N(b)|$ many clones remain unmatched in $N^*$. We match these clones to distinct dummy vertices in $\dL$. Since $N$ is critical matching $\Dfa{N}=\Dfa{M}=|\dLA|$ and $\Dfb{N}=\Dfb{M}=|\dLB|$. Thus, all the dummy vertices in $\dL=\dLA\cup \dLB$ get matched and hence $N^*$ is $\dL$-perfect matching. After matching these clones to dummy vertices, we say that none of the vertices remains deficient as at least $q^-(v)$ many positions of $v\in\AA\cup\BB$ are occupied by true or dummy vertices. For an under-subscribed $a$ and an under-subscribed $b$ in $N$, some clones $a_k$ and $b_k$ remains unmatched in $N^*$. We add $(a_k,\ell_{a_k})$ and $(b_k,\ell_{b_k})$ edges in $N^*$ to make $N^*$ an $(\AA' \cup \BB'\cup \dL)$-perfect matching. Below we give a formal description for constructing $N^*$ using $M$, $M^*$, $N$ and $G_M$. Initially, $N^*(v)=\bot$ for all $v\in\AA'\cup\BB'$.
      \begin{itemize}
          \item[(i)] For each edge $e=(a,b)\in M \cap N$, if $(a_i,b_j)\in M^*$ then we add the edge $(a_i,b_j)$ to $N^*$. That is, we set $N^*(a_i)=b_j$ and $N^*(b_j)=a_i$.
          \item[(ii)] For every edge $(a,b)\in N\setminus M$, we find appropriate clones $a_i$ of $a$ and $b_j$ of $b$ as given below and add $(a_i,b_j)\in N^*$. Note that for an edge $(a,b)\notin M$, $G_M$ has a complete bipartite graph between all clones of $a$ and all clones of $b$. Also, there must exist some unmatched clones of $a$ and unmatched clones of $b$ because we have $q^+(a)$ and $q^+(b)$ many clones of $a$ and $b$ respectively. While evaluating the votes, $a$ uses the correspondence function $\cor_a$. 
          \begin{itemize}
              \item If $\cor_a(b,N,M)=b'$ then $(a,b')\in M$ which in turn implies $(a_i,b'_k)\in M^*$ for some $i$ and $k$. To find the index $j$ we use correspondence function $\cor_b$. Let $\cor_b(a,N,M)=a'$ then $M^*$ contains an edge $(a'_p,b_r)$ for some $p$ and $r$. We let $j=r$ and include the edge $(a_i,b_j)$ to $N^*$. If $\cor_b(a,N,M)=\bot$ then we 
              choose a clone $b_k$ of $b$ such that $b_k$ is unmatched so far in $N^*$, that is, $N^*(b_k)=\bot$, and $M^*(b_k)=d\in \dL$. We add the edge $(a_i,b_k)$ in $N^*$. If no such $b_k$ exists then we choose a clone $b_k$ of $b$ such that $N^*(b_k)=\bot$ and $M^*(b_k)\in \tLB$. We add the edge $(a_i,b_k)$ to $N^*$.
              
              \item If $\cor_a(b,N,M)=\bot$ then we choose $a_i$ such that $N^*(a_i)=\bot$ and $M^*(a_i)\in \dLA$. If there is no such $a_i$ then we arbitrarily choose a clone $a_{i'}$ such that $M^*(a_{i'})\in \tLA$ and $N^*(a_{i'})=\bot$. Thus, we find the index $i=i'$. To determine the index $j$ we use  $\cor_b$. Let $\cor_b(a,N,M)=a'$ then $M^*$ contains an edge $(a'_p,b_r)$ for some $p$ and $r$. We consider $j=r$ and include the edge $(a_i,b_j)$ in $N^*$. If $\cor_b(a,N,M)=\bot$ then we add $(a_i,b_j)$ to $N^*$ for some $j$ such that $N^*(b_j)=\bot$ and $M^*(b_j)\in \dLB$. If there is no such $b_j$ then we arbitrarily choose a clone $b_k$ such that $N^*(b_k)=\bot$ and $M^*(b_k)\in \tLB$. We add  $(a_i,b_k)$ to $N^*$.
          \end{itemize}
          
          \item[(iii)] Consider a vertex $v\in\AA\cup\BB$ such that $v$ is deficient in $N$. At most $q^-(v)-|N(v)|$ many clones of $v$ can be left unmatched from $\AA'$ (or $\BB'$) in $G_M$ which are not adjacent to last-resorts. We select all the unmatched clones of $v$ which are not adjacent to any last-resorts and match them to arbitrary but distinct dummy vertices in $\dL$. We add the corresponding edges to $N^*$.
          
          \item[(iv)] For any vertex $v_k\in \AA' \cup \BB'$ that is left unmatched in the above step, we select an arbitrary but unmatched dummy vertex $d_j$ if it exists, and add the edge $(v_k,d_j)$ to $N^*$. If there does not exist any unmatched dummy then we select an arbitrary but distinct unmatched last-resort, say $\ell_{v_j}$,  if it exists, and add the edge $(v_k,\ell_{v_j})$ to $N^*$.
        \end{itemize} 
 Now we show that $N^*$ is an $(\AA'\cup\BB'\cup\dL)$-perfect one-to-one matching in $G_M$. Note that $N$ is a critical matching in $G$. If $|N(v)|\ge q^-(v)$ for a vertex $v\in \AA'\cup\BB'$ then we have sufficient last-resorts ($q^+(v)-q^-(v)$ many) to match all the clones of $v$. If $|N(v)|<q^-(v)$ for a vertex $v\in \AA'\cup\BB'$ then because of step (i) and (ii) at least $|N(v)|$ many clones of $v$ get matched to some true clones. We use similar argument as in the proof of Lemma~\ref{lem:crit2perf} to show that $N^*$ is an $(\AA'\cup\BB')$-perfect one-to-one matching. Note that $N$ and $M$ both are critical matchings. Thus, $\Dfa{N}=\Dfa{M}=|\dLA|$ and $\Dfb{N}=\Dfb{M}=|\dLB|$. Moreover, all dummy vertices in $\dLA$ are adjacent to all clones in $\AA'$, and all dummy vertices in $\dLB$ are adjacent to all clones in $\BB'$. First, we use dummy vertices to match clones in steps (iii) and (iv). Thus, last-resorts are used only when all dummy vertices are matched which implies that all the dummy vertices are also matched in $N^*$. Hence, each $v\in \AA'\cup \BB'\cup\dL$ are matched in $N^*$. It is also easy to see that $N^*$ is a one-to-one matching in $G_M$. 
 
 Next, we compute the weight of $N^*$ and show that it is $\Delta(N,M,\cor)$. Equation~(\ref{eqn:sumCorr}) below follows from Equation~(\ref{eqn:cloneCorr}) using the fact that the matching $N^*$ is an $(\AA'\cup\BB')$-perfect and each $v\in \AA\cup\BB$ has exactly $q^+(v)$ many clones matched in $N^*$.
\begin{eqnarray}
        wt(N^*) &=&  \sum_{e\in N^*}wt(e)\\ &=&\sum_{(a_i,b_j)\in N^*}wt(a_i,b_j) + \sum_{(v_k,\ell_{v_k})\in N^*}wt(v_k,\ell_{v_k}) + \sum_{(v_k,d_m)\in N^*}wt(v_k,d_m)\\
        & =&\sum_{(a_i,b_j)\in N^*}(vote_{a}(N^*(a_i),M^*(a_i)) + vote_{b}(N^*(b_j),M^*(b_j))) \nonumber\\ 
        & & \quad+ \sum_{(v_k,\ell_{v_k})\in N^*}vote_{v_k}(\ell_{v_k},M^*(v_k)) + \sum_{(v_k,d_m)\in N^*}vote_{v_k}(d_m,M^*(v_k))\label{eqn:cloneCorr}\\
  		& =& \sum_{v\in \AA\cup\BB} \sum_{i=1}^{q^+(v)} vote_v(N^*(v_i),M^*(v_i))\label{eqn:sumCorr}\\
        & =&\sum_{v\in \AA\cup\BB} vote_v(N,M,\cor_v)
        = \Delta(N,M,\cor)
\end{eqnarray} 
 \end{proof}

Note that $M$ is a critical matching (Lemma~\ref{lem:CriticalM}). Thus, if $\Delta(N,M,\cor)\le 0$ for every critical matching $N$, then matching $M$ is popular amongst critical matchings. Hence, it suffices to show that every $(\AA' \cup \BB'\cup \dL)$-perfect one-to-one matching in $G_M$ has weight at most $0$. 
To prove this, we write an LP for the \emph{maximum weight} $(\AA' \cup \BB'\cup \dL)$-perfect one-to-one matching and establish a dual feasible solution with value zero.

\subsection{Linear program and its dual}
Given the weighted graph $G_M$ we use the standard linear program (LP)  to compute a maximum weight $(\AA' \cup \BB'\cup\dL)$-perfect matching in $G_M$. 

The LP and its dual (dual-LP) are given below. For the (primal) LP we have a variable $x_e$ for every edge in $E'$. We let $\delta(v)$ denote the set of edges incident on the vertex $v$ in the graph $G_M$. 

\begin{eqnarray*}
\mbox{\bf LP:}  \hspace{0.5in}	&\max& \ \ \  \sum\nolimits_{e\in E'} wt(e) \cdot x_e \\ 
	\mbox{ subject to:}  & & \\   
	\sum\nolimits_{e\in \delta(v)}x_e &=& 1\ \ \ \forall v \in \  \AA'\cup \BB'\cup\dL\\
	 \sum\nolimits_{e\in \delta(\ell_{v})}x_e &\leq& 1 \ \ \ \forall \ell_{v} \in \tLA\cup\tLB \\
	 \ \ \ x_e &\geq& 0 \ \ \ \forall e\in E' \qquad
\end{eqnarray*}
We obtain the dual of the above LP by associating a variable $\alpha_v$ for every $v \in \AA' \cup \BB' \cup \tL\cup\dL$.

\begin{eqnarray}
\mbox{\bf dual-LP:} \hspace{0.05in}	&& \min \sum\limits_{v\in \AA'\cup \BB'\cup \dLA\cup \dLB \cup \tLA \cup \tLB} \alpha_{v}  \nonumber \\
	\mbox{subject to:}  \nonumber \\
  \alpha_{a_i} + \alpha_{b_j} &\ \geq\ & wt(a_i,b_j) \hspace{0.05in}\ \ \  \forall (a_i,b_j)\in E'\ \mbox{ where }  a_i \in \AA', b_j \in \BB' \label{eqn:dual1} \\%
	 \alpha_{\ell} + \alpha_{b_j} &\ \geq\ & wt(\ell,b_j) \hspace{0.05in}\ \ \ \ \  \forall (\ell, b_j) \in E'\ \mbox { where } \ell \in \tLB, b_j \in \BB' \label{eqn:dual2}\\
	 \alpha_{a_i} + \alpha_{\ell} &\ \geq\ & wt(a_i,\ell) \hspace{0.05in}\ \ \ \ \  \forall (a_i,\ell) \in E'\ \mbox{ where } \ell \in \tLA, a_j \in \AA' \label{eqn:dual3}\\
	 \alpha_{a_i} + \alpha_{d} &\ \geq\ & wt(a_i,d) \hspace{0.05in}\ \ \ \ \  \forall (a_i,d) \in E'\ \mbox{ where } d \in \dLA, a_i \in \AA' \label{eqn:dual4}\\
	 \alpha_{d} + \alpha_{b_j} &\ \geq\ & wt(d,b_j) \hspace{0.05in}\ \ \ \ \ \forall (d,b_j) \in E'\ \mbox{ where } d \in \dLB, b_j \in \BB' \label{eqn:dual5}\\
	 \alpha_{\ell} &\ \geq\ & 0 \hspace{0.7in} \ \ \ \forall \ell \in \tL \label{eqn:dual6}\ \   
\end{eqnarray}

\subsection{Dual assignment and its correctness} \label{sec:dfeas}
Now, we present an assignment of values to the dual variables of the dual-LP below. The dual assignment is shown in Figure~\ref{fig:maxlevelgraph} in blue-colored text. We prove in Lemma~\ref{lem:dSumZero} that the proposed dual assignment is feasible and the sum of the dual values is zero. 

\begin{itemize}
\item  Set $\alpha_{v}=0$ if $v\in \tL$ or $v\in \AA'\cup \BB'$ such that  $M^*(v)\in\tL$.
\item Set $\alpha_{a_i}=2x+1\ \ \ $ $\forall$ $a_i\in \AA'_{\T-x}$, where $-(\S+1)\le x \le \T$ and $M^*(a_i)\notin \tLA$.
\item Set $\alpha_{b_j}=-(2x+1)\ \ \ $ $\forall$ $b_j\in \BB'_{\T-x}$, where $-(\S+1)\le x \le \T$ and $M^*(b_j)\notin \tLB$.
\end{itemize}

\begin{lemma}\label{lem:dSumZero}
The above dual assignment is feasible, and the sum of the dual values is zero.
\end{lemma}

\begin{proof}
Recall that if an edge $(u,v)\in M^*$ then, by construction of $G_M$, $wt(u,v)=0$. We observe that our dual assignment has a property that, if $(u,v)\in M^*$ and one end point, say $u$, is assigned a dual value $\alpha_u=k$ then the other end point $v$ is assigned the dual value $\alpha_v=-k$. Therefore, with respect to dual inequality every matched edge satisfies the dual inequality. This immediately implies that the sum of dual values is zero because $M^*$ is an $(\AA'\cup \BB'\cup \dL)$-perfect matching and all last-resorts are assigned dual values equal to $0$. Hence, it follows that $\sum_{v\in \AA'\cup \BB'\cup \dL\cup\tL}\alpha_v=0$. This completes one part of the proof. In the rest of the proof we show that for each edge not in $M^*$ the respective dual inequalities are satisfied.

We note that if an edge is unmatched then it can be one of the three types (i) both endpoints are true clones, (ii)  one endpoint is a last-resort and, (iii) one endpoint is a dummy vertex. 

\noindent\textbf{Both the endpoints are true clones}: Here, we show that inequality~(\ref{eqn:dual1}) of the dual LP is satisfied for each edge $(a_i,b_j)\notin M^*$ for our dual assignments. Lemma~\ref{pr:noedge} states that there is no steep downward edge in $G_M$. Therefore, the edges in $G_M$ can be of three kinds -- one level downward,  belonging to the same level, or upward. We consider all these three cases below and use Lemma~\ref{lem:edgewts}.

\begin{itemize}
		\item[-] \textbf{One level downward edges}: Let $e$ be an edge such that $e\in \AA'_{\T-(x-1)}\times \BB'_{\T-x}$ where $-\S\le x\le \T$. By using Lemma~\ref{lem:edgewts}(\ref{itm:y1}), we know that such an edge has weight equal to $-2$. As per our dual assignment $\alpha_{a_i} +\alpha_{b_j}=(2(x-1)+1) + -(2x+1)=2x-2+1-2x-1=-2\geq wt(e)$. Thus, the dual inequality is satisfied.
		\item[-] \textbf{Edges belonging to the same level}: Let $e$ be an edge such that $e\in \AA'_{\T-x}\times \BB'_{\T-x}$  where $-(\S+1)\le x\le \T$. By using Lemma~\ref{lem:edgewts}(\ref{itm:y2}), we know that $wt(e)\le 0$.  As per our dual assignment $\alpha_{a_i} +\alpha_{b_j} = (2x+1) -(2x+1)=0\ge wt(e)$. Thus, the dual inequality is satisfied.
		\item[-] \textbf{Upward edges}: Let $e$ be an edge such that $e\in \AA'_{\T-x}\times \BB'_{\T-y}$ where $-(\S+1)\le x\le \T$ and $y<x$. By using Lemma~\ref{lem:edgewts}(\ref{itm:y3}), we know that $wt(e)\le 2$. As per our dual assignment $\alpha_{a_i} +\alpha_{b_j} = (2x+1) - (2y+1)=2(x-y)\ge 2 \ge wt(e)$. Thus, the dual inequality is satisfied.
	\end{itemize}

\noindent\textbf{One endpoint is a last-resort}: Now, we consider unmatched edges of $G_M$ whose one endpoint is a last-resort.
By construction we know that $wt(\ell,b_j)$ for $\ell\in\tLB$ is at most~0. Using Lemma~\ref{lem:noedgeLresS} we see that such $b_j$ can only be in $\BB'_x$ for $x\ge \T$. If $b_j\in \BB'_x$ for $x> \T$ then $\alpha_{b_j}$ is positive and we are done. Thus, for such an edge  Ineq.~(\ref{eqn:dual2}) is satisfied. If $b_j\in \BB'_{\T}$ then $b_j$ can be matched to a last-resort or to a true clone. If $b_j$ is matched to a last-resort, that is $(\ell',b_j)\in M^*$ for some $\ell'\in\tLB$, then $\alpha_{b_j}=0$  and $wt(\ell,b_j)=0$. Thus, for such an edge  Ineq.~(\ref{eqn:dual2}) is satisfied. If $b_j$ is matched to a true clone, that is $(a_i,b_j)\in M^*$ for some $a_i\in\AA'$, then $\alpha_{b_j}=-1$ and $wt(\ell,b_j)=-1$. Thus, Ineq.~(\ref{eqn:dual2}) is satisfied.

Now we consider the case where an edge is of the form $(a_i,\ell)$ where $\ell\in\tLA$. By construction of $G_M$, $wt(a_i,\ell)$ is at most~0  for all $\ell\in\tLA$. By Lemma~\ref{lem:noedgeLresC} we know that such $a_i$ can only be in $\AA'_x$ for $x\le \T+1$. If $a_i\in \AA'_x$ for $x\le \T$ then $\alpha_{a_i}$ is positive and we are done. Thus, for such an edge  Ineq.~(\ref{eqn:dual3}) is satisfied. So let us assume $a_i\in \AA'_{\T+1}$. Then $a_i$ can be matched to a last-resort or to a true clone. If $a_i$ is matched to a last-resort, that is $(a_i,\ell')\in M^*$ for some $\ell'\in\tLA$, then $\alpha_{a_i}=0$ and $wt(a_i,\ell)=0$. This implies that the  Ineq.~(\ref{eqn:dual3}) is satisfied.  If $a_i$ is matched to a true clone, that is $(a_i,b_k)\in M^*$ for some true clone $b_k\in\BB'$, then $\alpha_{a_i}=-1$ and $wt(a_i,\ell)=-1$.  Thus, for such an edge  Ineq.~(\ref{eqn:dual3}) is satisfied.

\vspace{0.1in}

\noindent\textbf{One endpoint is a dummy vertex}: Now, we consider unmatched edges of $G_M$ whose one endpoint is a dummy vertex. Let us first consider an edge $e$ of the form $(a_i,d)$ where $d\in\dLA$. Note that by the construction of $G_M$, $wt(e)\le 0$. Also note that by construction of $G_M$, all the dummy vertices in $\dLA$ are in $\BB'_{\S+\T+1}$.  This implies that all such edges are of the form $\AA'_{y}\times \BB'_{\S+\T+1}$ where $0\le y\le \S+\T+1$. As per our dual assignment, the minimum possible $\alpha$-value for a vertex in $\AA'_y$ for $y\le \S+\T+1$ is $-2\S-1$. Thus, $\alpha_{a_i} +\alpha_{d}\ge  -2\S-1 + (2\S+1)=0\ge wt(e)$. Thus, Ineq.~(\ref{eqn:dual4}) is satisfied.

Now, let us consider an edge $e$ of the form $(b_j,d)$ where $d\in\dLB$. By the construction of $G_M$, $wt(e)\le 0$. Also note that by construction of $G_M$, all the dummy vertices in $\dLB$ are in $\AA'_{0}$. This implies that all such edges are of the form $\AA'_{0}\times \BB'_{y}$ where $0\le y\le \S+\T+1$. As per our dual assignment, the minimum possible $\alpha$-value for a vertex in $\BB'_y$ for $y\le \S+\T+1$ is $-2\T-1$. Thus, $\alpha_{d} +\alpha_{b_j}\ge  2\T+1 -2\T-1=0\ge wt(e)$. Thus, Ineq.~(\ref{eqn:dual5}) is satisfied. 

 Ineq.~(\ref{eqn:dual6}) holds because all the last-resorts corresponding to a vertex are assigned $\alpha$-values equal to 0.
\end{proof}

 Lemma~\ref{lem:dSumZero} and the weak duality theorem together imply that the optimal value of the primal LP is at most 0. That is, every matching in $G_M$ that matches all vertices in $\AA'\cup \BB'\cup\dL$ has weight at most 0. Thus, by using Lemma~\ref{lem:CriticalM} and Lemma~\ref{lem:corr}, we establish that $M$ is a popular critical matchings. 
 
 Next, we observe one important property of our dual feasible solution $\Vec{\alpha}$ and the matching $M^*$ in Lemma~\ref{lem:dualSettingProp}. Given a dual feasible solution, we say that an edge $e$ is \emph{tight} if the dual inequality corresponding to $e$ is satisfied as equality. Otherwise, we say that the edge $e$ is a \emph{slack} edge.
 \begin{lemma}\label{lem:dualSettingProp}
  All the edges in $M^*$ are tight with respect to our dual feasible solution $\Vec{\alpha}$. Moreover, if a vertex $v$ is matched to a last-resort in $M^*$ then all the edges incident to $v$, except the edges incident to last-resorts, are slack.
 \end{lemma}
 
 \begin{proof}
 First we show that all the edges in $M^*$ are tight. Note that $wt(u,v)=0$ for an edge $(u,v)\in M^*$ where $u,v\in \AA'\cup \BB'\cup \dL\cup\tL$. It is clear from our dual assignment above that $\alpha_u=k$ and $\alpha_v=-k$ for the edge $(u,v)\in M^*$. This implies $\alpha_u+\alpha_v=0=wt(u,v)$. Thus, all the edges in $M^*$ are tight.
 
 Now we prove the next part of the claim. We split the proof based on if $v\in\AA'$ or $v\in\BB'$. Let us first consider $v=a_i\in\AA'$. That is, $a_i\in \AA'$ such that $M^*(a_i)\in \tLA$. Then by construction of $G_M$, $a_i\in \AA'_{\T+1}$. By using Lemma~\ref{pr:noedge}, we know that $a_i$ has no neighbour in $\BB'_x$ for $x\le \T-1$. Since $M^*(a_i)\in \tLA$, by Lemma~\ref{pr:vertexCloneG}(\ref{itm:partprop8}) we note that  $a_i$ cannot have neighbours in $\BB'_{\T}$ as well. This implies all such edges belong to the same level $\T+1$ or are upward. Consider an unmatched edge $(a_i,u)\in \AA'_{\T+1} \times \BB'_x$ for $x\ge\T+1$ such that $u\in \dLA\cup \BB'$. The edge $(a_i,u)$ incident on $a_i$ is a \emph{slack} because of the following reason. 
 \begin{itemize}
     \item If $u\in \BB'_{\T+1}$ then $u$ must be a true clone. This implies that $wt(a_i,u)=0$ (the vote by $b\in\BB$, for which $u$ is a clone, must be negative for $a$ otherwise $(a,b)\in  M$). As per our dual assignment $\alpha_{a_i}+\alpha_{u}= 0+1=1$.
     \item  If $u\in \BB'_x$ for $x\ge\T+2$ then $wt(a_i,u)\le 2$ whereas, $\alpha_{a_i}+\alpha_{u}\ge 0+3=3$.
 \end{itemize}

Now, let us assume that $v=b_j\in \BB'$. That is, $b_j\in\BB'$ such that $M^*(b_j)\in\tLB$.   Then by construction of $G_M$, $b_j\in \BB'_{\T}$. Since $M^*(b_j)\in\tLB$ the vertex $b$, for which $b_j$ is a clone, must be under-subscribed in $M$. We claim that $b_j$ has no neighbour in $\AA'_x$ for $x>\T$. Suppose there is a neighbour $a_k$ in $\AA'_x$ for $x>\T$ such that $a_k$ is a clone of $a\in\AA$. Since $b$ is under-subscribed in $M$, the vertex of $a$ must get matched to $b$ before raising its level above $\T$. Thus, consider an edge $(u,b_j)\in \AA'_x\times \BB'_{\T}$ for $x\le \T$ such that $u\in\dLB\cup\AA'$. The edge $(u,b_j)$  incident on $b_j$ is a \emph{slack} because of the following reason. 
\begin{itemize}
    \item If $u\in \AA'_{\T}$ then $u$ must be true clone. This implies that $wt(u,b_j)=0$ (the vote by $a\in\AA$, for which $u$ is a clone, must be negative for $b$ otherwise $(a,b)\in  M$). As per our dual assignment $\alpha_{a_i}+\alpha_{b_j}=1+0=1$.
    \item If $a_i\in \AA'_x$ for $x\le\T-1$ then $wt(a_i,b_j)\le 2$ whereas $\alpha_{a_i}+\alpha_{b_j}\ge 3+0= 3$.
\end{itemize}
This completes the proof of lemma.
 \end{proof}

 Using Lemma~\ref{lem:dualSettingProp}, we establish in Lemma~\ref{lem:slack} that if a clone is matched to a last resort in our matching $M^*$ then that clone remain matched to last resort in an $(\AA'\cup \BB'\cup\dL)$-perfect one-to-one matching $N^*$ corresponding to a popular critical matching $N$ in $G$. 
 \begin{lemma}\label{lem:slack}
Let $N$ be a popular critical matching in $G$. Let $N^*$ be the corresponding $(\AA'\cup \BB'\cup\dL)$-perfect one-to-one matching  in $G_M$ such that $wt(N^*)=\Delta(N, M,\cor)$. If $M^*(v)\in\tL$ for a  vertex $v\in (\AA'\cup \BB')$ then $N^*(v)\in\tL$. 
\end{lemma}

\begin{proof}
 Using Lemma~\ref{lem:dSumZero}, we conclude that any dual feasible solution with dual-sum value of zero is an optimal dual solution. Thus, the dual values assigned to vertices in $\AA'\cup \BB'\cup \dL\cup\tL$ is an optimal dual solution. The fact that $M$ is a popular critical matching and $N$ is a critical matching in $G$  implies $wt(N^*)=\Delta(N,M,\cor)\le 0$. Since $N$ is also popular critical matching, $wt(N^*)=0$ and thus $N^*$ is a primal optimal solution. This implies the complementary slackness conditions hold. Hence, for each edge $(a_i,b_j)$ in $G_M$ either $\alpha_{a_i}+\alpha_{b_j}=wt(a_i,b_j)$ or $(a_i,b_j)\notin N^*$. 

Let us assume that $v\in \AA'\cup \BB'$ such that $M^*(v)\in \tLA\cup\tLB$. Applying Lemma~\ref{lem:dualSettingProp}, we know that all the edges incident to $v$, except the edges incident to last-resorts, are \emph{slack} and hence, these edges are not in $N^*$. But $N^*$ is an $\AA'\cup\BB'$-perfect matching implying that if $v\in\AA'$ then $N^*(v)\in \tLA$ and, if $v\in\BB'$ then $N^*(v)\in \tLB$. 
\end{proof}
 
 \subsection{Maximum size popular matching amongst critical matching}\label{lem:maxcardpop}
 
 In this section, we show that $ M$ is a maximum cardinality matching amongst all the popular critical matchings, which establishes Theorem~\ref{theo:main}. We use Lemma~\ref{lem:slack} to show that $ M$ is a maximum cardinality matching amongst all the popular critical matchings. We also prove Theorem~\ref{theo:ruralHos} and Theorem~\ref{theo:2by3}.  

\begin{lemma}\label{lem:maxSize}
 The matching $M$ is a maximum cardinality popular critical matching in $G$.
\end{lemma}

\begin{proof}
Consider any critical matching $N$ in $G$ such that $|N|>|M|$. Now consider the corresponding one-to-one $(\AA'\cup\BB'\cup\dL)$-perfect matching $N^*$ in the graph $G_M$ such that $wt(N^*)=\Delta(N, M,\cor)$. Since $|N|>|M|$ and $N^*$ and $M^*$ both are $(\AA'\cup\BB'\cup\dL)$-perfect matching there must exists a clone $v_j$ such that $N^*(v_j)\notin\tL$ but $M^*(v_j)\in\tL$. By Lemma~\ref{lem:slack} we claim that $N^*$ uses a slack edge. We know that the dual-feasible solution $\Vec{\alpha}$ is an optimal solution to the dual LP.  The feasible solution corresponding to the  $(\AA'\cup\BB'\cup\dL)$-perfect matching $N^*$ cannot be an optimal solution to the primal LP because it contains a slack edge. This implies $wt(N^*)<0$ because the optimal value of the primal LP is 0. Thus, a critical matching $N$ such that $|N|>|M|$ loses to $M$. Hence, $M$ is a maximum size popular critical matching as no larger size critical matching is more popular than it. 
\end{proof}

Lemma~\ref{lem:CriticalM}, Lemma~\ref{lem:dSumZero} and Lemma~\ref{lem:maxSize} together establish Theorem~\ref{theo:main}. We use following lemma to prove Threorem~\ref{theo:ruralHos} and Theorem~\ref{theo:2by3}. 

\begin{lemma}\label{lem:dummy}
Let $N^*$ be an $(\AA'\cup \BB'\cup\dL)$-perfect one-to-one matching in the graph $G_M$ corresponding to a maximum size popular critical matching $N$ in $G$. Furthermore, assume that $wt(N^*)=\Delta(N, M,\cor)$. If $M^*(u)\in\dL$ for a  vertex $u\in (\AA'\cup \BB')$ then $N^*(u)\in\dL$. 
\end{lemma}

\begin{proof}
For the sake of contradiction let us assume that $M^*(u)=d\in \dL$ but $N^*(u)\notin\dL$. Without loss of generality, assume that $u\in\AA'$. 
By Lemma~\ref{lem:noedgeLresC} we know that $u$ is not adjacent to any last-resort. Using Claim~\ref{cl:property} (below) we show that if such a clone $u$ (that is, $u\in\AA'$ such that $M^*(u)=d\in \dL$ but $N^*(u)\notin\dL$) exists  
then there exists an alternating cycle $\rho$ in  $M^*\oplus N^*$ containing $u$ and another vertex $v_j\in\AA'$ such that $N^*(v_j)=d'\in \dL$ and $M^*(v_j)\in \BB'$. Assuming the existence of such a vertex $v_j$ we will show that the weight of such a cycle in $G_M$ is negative. This will imply that $M$ defeats $N$ in terms of votes along this cycle which contradicts the popularity of $N$.  Now we show that the weight of the cycle $\rho$ is negative.

Since $M^*(v_j)\in \BB'$, the vote by $v_j$ for the edge $(v_j,d')$ must be $-1$. Thus $wt(v_j,d')=-1$. Since $M^*(u)\in\dL$, by construction, $u\in\AA'_{\S+\T+1}$. 
If $\rho$ is restricted to $\AA'_{\S+\T+1}\times \BB'_{\S+\T+1}$ then by Lemma~\ref{lem:edgewts}(\ref{itm:y2}) we know that $\rho$ does not contain any positive weight edge. But the presence of the edge $(v_j,d')$ in $\rho$ ensures that the total weight of $\rho$ is negative. On the other hand, if $\rho$ uses edges other than the edges in $\AA'_{\S+\T+1}\times \BB'_{\S+\T+1}$ then let us assume that it uses $k$ downward edges. Since $\rho$ is a cycle and $G_M$ does not contain steep downward edges (Lemma~\ref{pr:noedge}) it must contain at least one but not more than $k$ upward edges. Recall that the edges in $G_M$ are implicitly directed from left to right and hence edges from $\AA'_{x}\times\BB'_{x-1}$ are called downward edges and edges from $\AA'_{x}\times \BB'_{y}$ for $x>y+1$ are called steep downward edges.

All the downward edges have weight equal to $-2$ (Lemma~\ref{lem:edgewts}(\ref{itm:y1})). All the upward edges have weight at most $2$ (Lemma~\ref{lem:edgewts}(\ref{itm:y3})). Note that $\rho$ contains the edge $(v_j,d')$ having weight equal to $-1$. 
Thus, $\rho$ contains at least $k$ edges other than $(v_j,d')$ each with weight exactly equal to $-2$ and at most $k$ edges each with weight $\le 2$. Therefore, the presence of a negative weight edge $(v_j,d')$ ensures that the weight of $\rho$ is negative. 
This implies that $N$ is not a popular critical matching, a contradiction. This completes the proof of Lemma~\ref{lem:dummy} except for the proof of Claim~\ref{cl:property} which we prove below.
 \end{proof}

\begin{cl}\label{cl:property}
 Let $N^*$ be an $(\AA'\cup \BB'\cup\dL)$-perfect one-to-one matching corresponding to a maximum size popular critical matching $N$ in the graph $G_M$ such that $wt(N^*)=\Delta(N, M,\cor)$. Let $u\in\AA'$ be such that $M^*(u)=d\in\dL$ but $N^*(u)\notin\dL$. Then we have an alternating cycle $\rho$ in $M^*\oplus N^*$ containing $u,N^*(u)$ and $d$.  Moreover, the cycle $\rho$ contains a vertex $v_j\in\AA'$ such that $M^*(v_j)\in\BB'$ and  $N^*(v_j)=d'\in \dL$.
\end{cl}
\begin{proof}
Here we will show the existence of an alternating cycle containing $u, N^*(u)$ and $d$. For the sake of convenience, let us assume that $u=g_0$, $d=h_0$ and $N^*(g_0)=h_1$.  Observe that $\rho'=\langle h_0, g_0,h_1\rangle$ is an alternating path containing vertices $d, u$ and $N^*(u)$. We extend $\rho'$ to get an alternating cycle $\rho=\langle h_0,g_0,h_1,g_1\ldots,h_{i-1},g_{i-1},\ldots, h_k=h_0\rangle$ where $(h_j,g_j)\in M^*$ for $0\le j\le k-1$ and other edges are in $N^*$. 
We claim that $\rho$ does not contain any last resort.  This will eventually imply that $\rho$ must terminate by visiting one of the already visited vertices. Note that the degree of a vertex in $M^*\oplus N^*$ is at most two, and hence $\rho$ must revisit $h_0=d$ completing the cycle.  Thus, we have an alternating cycle containing $d,u, N^*(u)$.

First, we claim that no $g_j$ on $\rho$ is in $\tL$. Suppose for some $0\le j\le k-1$, $g_j\in\tL$. We apply Lemma~\ref{lem:slack} on $v=h_j$ to conclude that $g_{j-1}$ is last-resort. We repeatedly apply this to conclude that $g_0\in\tL$ which contradicts the fact that a last-resort $g_0$ is adjacent to a dummy vertex $h_0$ in $G_M$. Thus, no $g_j$ on $\rho$ is in $\tL$.

Next, we show that no $h_i$ on $\rho$ is in $\tL$. Assume for the sake of contradiction that for some $i$, $h_i\in \tL$.  We consider the following two cases depending on whether $h_i$ is matched in $M^*$ or not-- 

\begin{figure}
\begin{center}
    \scalebox{0.85}{\input{claimFig3}}
\end{center}

\caption{Alternating path $\rho$ in the case where $h_i\in\tL$ is unmatched in $M^*$. Solid (black/blue/gray) edges denote the edges in $M^*$ whereas dashed (black/blue/gray) edges denote the edges in $N^*$. The two sub-paths $\rho_1$ and $\rho_2$ of $\rho$ are shown in black and blue color, respectively. The vertices shown as black squares are in $\AA'_{\S+\T+1}\cup\BB'_{\S+\T+1}$. The red vertices denote last-resorts. By the construction of $G_M$, $h_i\in \BB_{\T+1}'$ and $g_{i+1}'\in\AA_{\T}'$. }
\label{fig:claim}
\end{figure}

\begin{enumerate}[(i)]
    \item \textbf{$h_i$ is unmatched in $M^*$:} In this case, the path $\rho$ must end at $h_i$. 
    Therefore, we extend $\rho$ towards the other side of the vertex $h_0$ to find the other endpoint of the path $\rho$ (see the blue sub-path in Figure~\ref{fig:claim}). The idea is to show that the other endpoint is also a last-resort which allows us to establish that $N$ is not a popular critical matching, a contradiction. We extend $\rho$ as follows. First, we include the edge $(h_0,N^*(h_0))$ to $\rho$. For convenience, let us assume that $N^*(h_0)=g'_1$. We remark that $g'_1$ is a true clone in $\AA'$ because $h_0$ is a dummy vertex and a dummy vertex is not adjacent to any other dummy vertex or last-resort. Let us consider $M^*(g_1')=h_1'$. Note that $h_1'\notin\tL$, otherwise applying Lemma~\ref{lem:slack} on $v=g_1'$ we get that $h_0\in \tL$, a contradiction to the fact that $h_0\in\dL$. Thus, the path $\rho$ cannot end at $h_1'$.  In a similar manner, we argue that that no $h'_i$ on $\rho$ is in $\tL$.  Therefore, we extend this path using alternate edges of $M^*$ and $N^*$, and conclude that $\rho$ must end at a vertex $g_{i+1}'$ such that $g_{i+1}'\in\tL$. Thus, the path $\rho$ has both its endpoints $h_i$ and $g_{i+1}'$ as last-resorts. 
 
 We observe the following about the levels of the vertices of the path $\rho$. The fact that $h_0\in \dL$ and $M^*(h_0)=g_0$ is a true clone implies that  $g_0\in \AA'_{\S+\T+1}$ and $h_0\in \BB'_{\S+\T+1}$. By Lemma~\ref{pr:vertexCloneG}(\ref{itm:partprop7}) we know that each neighbour of $g_0$, in particular, $h_1$ is in $\BB'_{\S+\T+1}$. Thus, by the construction of $G_M$, $g_1\in \AA'_{\S+\T+1}$. Since $h_i$ and $g_{i+1}'$ are last-resorts, by the construction of $G_M$, $h_i\in \BB'_{\T+1}$ and $g_{i+1}'\in\AA'_{\T}$. Applying Lemma~\ref{lem:noedgeLresC}, we claim that $g_{i-1}\in \AA'_{x}$ for $x\le \T+1$. Similarly, by applying Lemma~\ref{lem:noedgeLresS}, we claim that $h_i'\in \BB'_{x}$ for $x\ge \T$.

Now we show that the matching $M^*$ gains more votes than the matching $N^*$ along the path $\rho$ which establishes that $N^*$ is not a popular critical matching. To see this, consider the two sub-paths $\rho_{1}=\langle g_0,h_1,\ldots,g_{i-1}\rangle$ and $\rho_{2}=\langle g_1',h_1'\ldots,h_i'\rangle$ of $\rho$. In  Figure~\ref{fig:claim}, the sub-path $\rho_1$ is shown in black color and the sub-path $\rho_2$ is shown in blue color. Note that the extreme edges incident to last-resorts and the two other edges $(g_0,h_0)$ and $(h_0,g_1')$ shown in gray color in the figure are not the part of either of $\rho_1$ and $\rho_2$. Let $z_1$  and $z_2$ denote the following.

$$ z_1=\# \mbox{downward\ edges\ in\ }\rho_1 - \# \mbox{upward\ edges\ in\ } \rho_1$$
$$z_2=\# \mbox{upward\ edges\ in\ } \rho_2 - \# \mbox{downward\ edges\ in\ } \rho_2$$

Recall that in $G_M$ there is no \emph{steep} downward edge but the upward edges can be steep. Since $g_0\in \AA'_{\S+\T+1}$ and $h_i\in\BB'_{\T+1}$, the value of $z_1$ must be at least $\S$. Similarly, since the highest possible level for $h_1$ is $\S+\T+1$ and the lowest possible level of $h_i'$ is $\T$, the value of $z_2$ is at most $\S+1$.

By Lemma~\ref{lem:edgewts}(\ref{itm:y1}) we know that the weight of each downward edge is $-2$ and Lemma~\ref{lem:edgewts}(\ref{itm:y3}) implies that the weight of each upward edge is at most $2$. Thus, the total weight of $\rho_1\cup\rho_2$ is at most $+2$. Now, observe that $wt(g_1',h_0)=-1$ because $M^*(g_1')=h_1'$ is a true clone whereas $h_0$ is a dummy vertex. Also, the two extreme edges of $\rho$, $(h_i',g_{i+1}')$ and $(g_{i-1},h_i)$ have weight equal to $-1$ because the corresponding $N^*$-partner is a last-resort but $M^*$-partner is a true clone. Thus, the total weight of $\rho$ is at most $-1$. This implies that $M^*$ gains at least one vote along $\rho$. Thus, switching along $\rho$ gives us a more popular matching than $N$. Therefore, $N$ is not a popular critical matching, a contradiction. Thus, such a path $\rho$ which ends at $h_i\in\tL$ does not exist.

\item \textbf{$h_i$ is matched in $M^*$: }Recall that by the construction of $G_M$ a last-resort is adjacent to only true clones. Since $h_i\in\tL$, $M^*(h_i)=g_i$ is a true clone. By using Lemma~\ref{lem:slack} on the vertex $v=g_i$ we conclude that $h_{i+1}$ is in $\tL$. We continue in this way until we find some $h_j\in \tL$ such that $h_j$ is unmatched in $M^*$. We must find such $h_j$ because we have finite number of vertices in $G_M$ and all $h_p$ for $p\ge i$ is in $\tL$ (using Lemma~\ref{lem:slack} on vertex $v=g_{p}$). This implies that the path $\rho$ ends with an $N^*$-edge at a last-resort, and we get into the previous case.
\end{enumerate}

Thus, we conclude that no $h_i$ on $\rho$ is a last-resort. This completes the proof of the first part of the claim that we have an alternating cycle containing $d,u, N^*(u)$. 


To prove the other part of the claim, we need to show that the cycle $\rho$ contains a vertex $v_j\in \AA'$ such that $M^*(v_j)\in\BB'$ and  $N^*(v_j)=d'\in \dL$. 
We have already established that $\rho=\langle d=h_0,u=g_0,h_1,g_1\ldots,h_{k-1},g_{k-1},h_k=h_0\rangle$ such that no vertex on $\rho$ is a last-resort. By assumption, $N^*(u)=h_1\notin \dL$. Thus, it must be the case that $h_1\in \BB'$. This implies that $\rho$ starts with a dummy vertex $h_0$ and the next vertex $h_1\in \BB'$. Since $\rho$ is a cycle there must exists an index $p$ such that $0\le p\le k-1$, $h_p\in \BB'$ and $h_{p+1}\in \dL$ (sum is considered as modulo $k$). Now consider $g_p=M^*(h_p)$. Note that $M^*(g_p)\in\BB'$ and $N^*(g_p)=h_{p+1}\in \dL$. Thus, $g_p$ is the desired vertex $v_j$.
 \end{proof}

\begin{proofof}{Theorem~\ref{theo:ruralHos}}
Let $N$ be an arbitrary maximum cardinality popular critical matching in the instance $G$. Consider the $(\AA'\cup \BB'\cup \dL)$-perfect  one-to-one matching $N^*$ corresponding to $N$  in the graph $G_M$ such that $wt(N^*)=\Delta(N, M,\cor)$. Since $N$ is a popular critical matching, Lemma~\ref{lem:slack} holds. Thus, for any clone $u\in\AA'\cup\BB'$ if $M^*(u)\in\tL$ then $N^*(u)\in\tL$. Now we apply Lemma~\ref{lem:dummy} to claim that $M^*(u)\in \dL$ implies $N^*(u)\in\dL$. Thus, if $|M(v)| = k$ for any $v \in\AA\cup\BB$ then $|N(v)| = k$. This completes the proof of Theorem~\ref{theo:ruralHos}.
 \qed\end{proofof}

\begin{proofof}{Theorem~\ref{theo:2by3}}
First, note that there must exist a maximum cardinality matching which is also critical (if a critical matching is not a maximum cardinality matching, then we can always augment it to obtain a larger size critical matching). Hence, let us assume that $M_{max}$ is a maximum size critical many-to-many matching. Recall that the many-to-many matching $M$ given by our algorithm is converted into a one-to-one matching $M^*$ (in $G_M$) which is further converted into a matching $M'=M^*\setminus\{(v_k,\ell_{v})\ :\ v_k\in\AA'\cup\BB' \mbox{ and  }\ell_v\in\tL\cup\dL\}$ such that the one-to-one matching $M'$ does not contain any last-resort or dummy vertex (as done in Section~\ref{subsec:critical}). Thus $|M'|=| M|$. This process can be done for any critical matching. Thus, corresponding to the maximum size critical many-to-many matching $M_{max}$ we have a one-to-one matching without any vertices matched to last-resorts or dummy vertices, denoted by $M_{max}'$. Thus, $|M_{max}|=|M'_{max}|$. Now we show that there is no 1-length or 3-length augmenting path with respect to $M'$ in $M'\oplus M_{max}'$ which will imply that $|M'|\ge \frac{2}{3}\cdot |M_{max}'|$.

By Claim~\ref{cl:fully-subscribed}, our matching $M$ is maximal, and hence a 1-length augmenting path with respect to $M'$ does not exist. Now let us consider a 3-length augmenting path. Let $\langle u_1,u_2,u_3,u_4\rangle$ be a 3-length augmenting path  with respect to $M'$ such that $(u_1,u_2)$ and $(u_3,u_4)$ are in $M_{max}'$, and $(u_2,u_3)$ is in $M'$. Without loss of generality, assume that $u_1\in \AA'$ and $u_4\in\BB'$. Since $u_1$ is unmatched in $M'$ and all the unmatched vertices of $\AA'$ are in $\AA'_{\T+1}\cup \AA'_{\S+\T+1}$, $u_1$ must be in $\AA'_{\T+1}\cup \AA'_{\S+\T+1}$. Similarly, since $u_4$ is unmatched and all the unmatched vertices of $\BB'$ are in $\BB'_{\T}\cup \BB'_{0}$ it must be the case that $u_4\in \BB'_{\T}\cup \BB'_{0}$. Recall that $G_M$ does not contain any steep downward edge and for any unmatched vertex $v\in\AA'_{\S+\T+1}$ all the neighbours of $v$ are only in $\BB'_{\S+\T+1}$. Thus, it must be the case that $u_1\in \AA'_{\T+1}$ and $u_4\in \BB'_{\T}$, otherwise the length of the path will be more than three. Using Lemma~\ref{pr:vertexCloneG}(\ref{itm:partprop8}) we claim that $u_2\in \BB'_x$ for $x\ge \T+1$. Since $(u_2,u_3)\in M'$, $u_3\in \AA'_x$ for $x\ge \T+1$. Let $u_3$ be the clone of a vertex $a_1\in\AA$ and $u_4$ be the clone of a vertex $b_1\in\BB$. Then $a_1^t$ must have proposed to $b_1$ during the course of Algorithm~\ref{algo:maxPopSC2LQ}. Note that $b_1$ can reject $a_1$ only if it is fully subscribed which implies all the clones of $b_1$ (including $u_4$) are matched in $M'$. This is a contradiction to the fact that $u_4$ is an unmatched clone in $M'$.
\qed\end{proofof}

%% file: claimFig3.tex
\begin{tikzpicture}[thick,
  lnode/.style={draw=white,fill=black,scale=0.7},tnode/.style={draw=white,fill=red,scale=0.7},
  fsnode/.style={draw,circle,fill=black,scale=0.5},
  every fit/.style={ellipse,draw,inner sep=-2pt,text width=0.5cm},
  -,shorten >= 3pt,shorten <= 3pt]

\node at (-0.2,0) {$\AA'\cup\dLB\cup\tLB$};
\node at (3.2,0) {$\BB'\cup\dLA\cup\tLA$};


  \node[lnode] (g0) at (0,8) {};  \node at (-0.8,8) {$u=g_{0}$};
  \node[lnode] (h0) at (3,8) {};  \node at (3.8,8) {$h_{0}=d$};
  \node[lnode] (g1) at (0,6) {}; \node at (-0.6,6) {$g_{1}$};
  \node[lnode] (h1) at (3,6) {};  \node at (3.6,6) {$h_{1}$};
  \node[fsnode] (h2) at (3,4) {}; \node at (3.6,4) {$h_{2}$};
  \node at (0.4,4) {\ldots}; \node at (2.7,4) {\ldots};
  \node[fsnode] (gi2) at (0,4) {}; \node at (-0.6,4) {$g_{i-2}$};
  \node[fsnode] (hi1) at (3,2) {}; \node at (3.6,2) {$h_{i-1}$};
  \node[fsnode] (gi1) at (0,2) {}; \node at (-0.6,2) {$g_{i-1}$};
  \node[tnode] (hi) at (3,1) {}; \node at (3.6,1) {$h_{i}$};

  \draw[ultra thick,lightgray](g0) -- (h0);
  \draw[ultra thick,black,dashed](g0) -- (h1);
  \draw[ultra thick,black](g1) -- (h1);
  \draw[ultra thick,black,dashed](g1) -- (h2);
  \draw[ultra thick,black,dashed](gi2) -- (hi1);
  \draw[ultra thick,black](gi1) -- (hi1);
  \draw[ultra thick,lightgray,dashed](gi1) -- (hi);

 \node[fsnode] (g1') at (0,7) {};  \node at (-0.8,7) {$g_{1}'$};
  \node[fsnode] (h1') at (3,7) {};  \node at (3.8,7) {$h_{1}'$};
  \node[fsnode] (g2') at (0,5) {}; \node at (-0.6,5) {$g_{2}'$};
  \node[fsnode] (hi1') at (3,5) {}; \node at (3.6,5) {$h_{i-1}'$};
  \node[fsnode] (gi') at (0,3) {}; \node at (-0.6,3) {$g_{i}'$};
  \node[fsnode] (hi') at (3,3) {}; \node at (3.6,3) {$h_{i}'$};
  \node[tnode] (gi1') at (0,1) {}; \node at (-0.6,1) {$g_{i+1}'$};
  \node[blue] at (0.4,5) {\ldots};  \node[blue] at (2.7,5) {\ldots};
  \draw[ultra thick,lightgray,dashed](h0) -- (g1');
  \draw[ultra thick,blue](g1') -- (h1');
  \draw[ultra thick,blue,dashed](g2') -- (h1');
  \draw[ultra thick,blue,dashed](gi') -- (hi1');
  \draw[ultra thick,blue](gi') -- (hi');
  \draw[ultra thick,lightgray,dashed](gi1') -- (hi');
\end{tikzpicture}

%% file: 6conclusion.tex
\section{Conclusion}\label{sec:conclusion}
In this paper, we studied the many-to-many matching problem with two-sided preferences and two-sided lower quotas.  We showed that a critical matching that is popular amongst all critical matchings always exists. We presented a polynomial time algorithm to compute a maximum cardinality such matching. We also proved a variant of the Rural Hospital theorem where we showed that each vertex matches to the same capacity in every maximum cardinality popular critical matching. Finally, we showed that ignoring lower quotas and popularity conditions only allow at most 50\% more edges. Thus, we extend similar results proved in~\cite{kavitha2014size,NN17,brandl2019two,Kavitha2021,NNRS21} to the general setting.

%% file: 8appendix.tex
\newpage

\section{Challenges with two-sided lower quotas in many-to-many setting}\label{sec:challenges}
 Here, we comment on the possibility of extending two recent algorithms ~\cite{Kavitha2021,NNRS21} for popular matchings with two-sided lower quotas. Kavitha~\cite{Kavitha2021} considers popular matchings in one-to-one setting and Nasre \etal~\cite{NNRS21} consider popular matchings in many-to-one setting. Both these algorithms propose a reduction where the original instance with two-sided lower quotas is converted into an instance {\em without} lower quotas. The standard Gale-Shapley algorithm is used to obtain a  stable matching in the reduced instance and the edges in the stable matching are mapped to the edges in the original instance.
 The reductions used in the two works differ in the following aspects: -- the reduction in~\cite{Kavitha2021} makes copies of vertices only on one side of the bipartition. That is, vertices in $\AA$ have copies whereas the vertices in $\BB$ do not have copies and this is sufficient in  the one-to-one setting.  On the other hand, the reduction in~\cite{NNRS21} creates copies for vertices in both the partitions. Now we illustrate the difficulties in extending these approaches for many-to-many setting.

  Let $\S$ and $\T$ denote the sum of lower quotas of all the vertices in $\AA$ and $\BB$, respectively.  Let $\AALQ$ and $\BBLQ$ denote the set of \emph{lq}-vertices $a\in\AA$ and $b\in\BB$, respectively. In the reduced instance constructed in~\cite{Kavitha2021}, $\S+\T+2$ copies are created for each vertex in $\AALQ$ and $\T+2$ copies are created for each vertex in  $\AA\setminus\AALQ$ -- the $x$-th copy of a vertex $a$ is denoted as a level-$x$ copy of $a$. In contrast, every vertex in the set $\BB$ has exactly one copy in the reduced instance. Corresponding to each vertex in $\AALQ$ and $\AA\setminus\AALQ$, respectively $\S+\T+1$ and $\T+1$ many dummy vertices are introduced. Thus, the vertex set of the reduced instance, say $H=(\AA'\cup\BB',E')$ is as follows. The set $\AA'$ contains $\S+\T+2$ copies for each $a\in\AALQ$ and $\T+2$ copies of each $a\in\AA\setminus\AALQ$. The set $\BB'$ contains the vertex set $\BB$ along with the dummy vertices corresponding to all the vertices in $\AA$. 
  Preference lists of certain lower level-copies of a vertex in $\AA$ contains only the \emph{lq}-vertices in its original preference list. Preference list of a vertex in the reduced graph is obtained by suitably modifying the original preference list.

In the many-to-many setting each vertex has an associated lower and upper quota. One natural generalization would be to consider the quotas of some copies of each vertex equal to the lower quota of the corresponding vertex and quotas of remaining copies of each vertex equal to the upper quota of the corresponding vertex. For example -- let us consider a vertex $a\in\AA$. The quota of the copy $a^{x_1}$ for $x_1\le\T+1$ can be equal to $q^+(a)$ and that of the copy $a^{x_2}$ for $x_2\ge \T+2$ can be equal to $q^-(a)$. However, since we have only one copy for each $b\in\BB$, it is not immediate how to set the quota for the vertex. Setting its quota equal to the lower quota eliminates several critical matchings in the instance; on the other hand setting it equal to the upper-quota does not allow us to distinguish between the case when the upper quota was equal to the lower quota. This difficulty suggests that we need to have multiple copies of vertices in $\BB$ as well. Extending the reduction using this idea has the same challenge as we face in extending the reduction proposed in \cite{NNRS21}.

Now let us consider the reduction used in~\cite{NNRS21}. We directly give an overview of the modification of that reduction, which works for \emph{maximum} cardinality popular feasible matching (a popular feasible matching is a popular matching amongst all the matchings with deficiency 0). In this reduction, $\S+1$ and $\T+2$ level-copies are created for each vertex in $\AALQ$ and $\BBLQ$, respectively. Only one (resp. two) level-copy is created for a vertex in $\AA\setminus\AALQ$ (resp. $\BB\setminus\BBLQ$). Note that only $\T+1$ and one level-copies were created in~\cite{NNRS21} for vertices in $\BBLQ$ and $\BB\setminus \BBLQ$, respectively. Here, one extra copy for each vertex in $\BB$ is created because we are considering the \emph{maximum size} popular feasible matching. Also, $\S$ and $\T+1$ many dummy vertices are introduced corresponding to each vertex in $\AALQ$ and $\BBLQ$, respectively. The quota of the first level-copy of a vertex $a\in\AA$ is equal to $q^+(a)$, the quotas of first two level-copies of $b\in\BB$ are equal to $q^+(b)$,  the quotas of other level-copies of $v\in\AA\cup\BB$ are equal to $q^-(v)$. Preference lists are defined such that the higher level-copies appear before any lower level-copy of any vertex and, at the same level, the relative ordering of vertices is preserved.  

This reduction can be considered as it is for a many-to-many setting for computing a popular feasible matching. Note that feasible matching is a special case of critical matching. But for a many-to-many instance, the reduction does not have any mechanism to restrict the matching such that the two different level-copies of the same vertex do not get matched to same/different level-copies of another vertex. For example, let us consider two vertices $a\in\AA$ and $b\in\BB$ such that $q^+(a)-q^-(a)$ and  $q^+(b)- q^-(b)\ge 2$. Then the stable matching in the reduced instance may contain edges $(a^{x_1},b^{y_1})$ and $(a^{x_2},b^{y_2})$. Note that $a^{x_1}$ and $a^{x_2}$ are two level-copies of $a$ and $b^{y_1}, b^{y_2}$ are two level-copies of $b$. So $a$ is matched to $b$ at least twice, and it is counted towards the quotas of these level-copies. In fact, the example shown in Figure~\ref{Fig:CounterEx} is such an example.

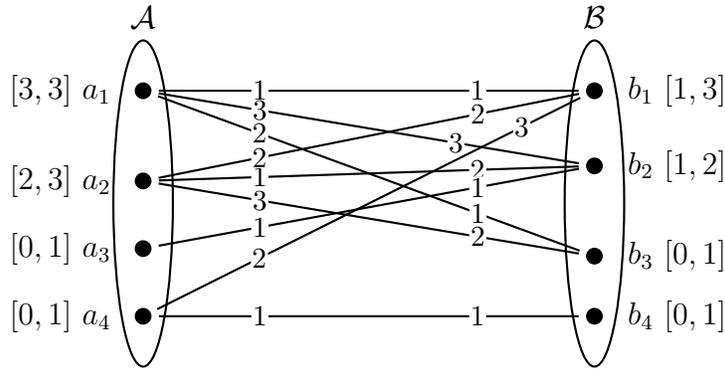
\begin{figure}[ht]
	\centering
	\begin{tikzpicture}[thick,
  lnode/.style={draw=white,fill=white},
  fsnode/.style={draw,circle,fill=black,scale=0.5},
  every fit/.style={ellipse,draw,inner sep=-2pt,text width=0.7cm},
  -,shorten >= 3pt,shorten <= 3pt
]

  \node[fsnode] (a1) at (0,3) {};
  \node[fsnode] (a2) at (0,1.8) {};
  \node[fsnode] (a3) at (0,0.9) {};
  \node[fsnode] (a4) at (0,0) {};
  \node[fsnode] (b1) at (6,3) {};
  \node[fsnode] (b2) at (6,2) {};
  \node[fsnode] (b3) at (6,0.8) {};
  \node[fsnode] (b4) at (6,0) {};
\node at (-1.1,3) {$[3,3]\ a_1$};
\node at (-1.1,1.8) {$[2,3]\ a_2$};
\node at (-1.1,0.9) {$[0,1]\ a_3$};
\node at (-1.1,0) {$[0,1]\ a_4$};
\node at (7.1,3) {$b_1\ [1,3]$};
\node at (7.1,2) {$b_2\ [1,2]$};
\node at (7.1,0.8) {$b_3\ [0,1]$};
\node at (7.1,0) {$b_4\ [0,1]$};

\node [fit=(a1) (a4),label=above:$\AA$] {};
\node [fit=(b1) (b4),label=above:$\BB$] {};

\draw(a1) -- (b1) node[pos=0.25,draw=white, fill=white,inner sep=0.1pt] {\textcolor{black}{\small 1}} node[pos=0.75,draw=white, fill=white,inner sep=0.1pt] {\textcolor{black}{\small 1}};
\draw(a1) -- (b3) node[pos=0.25,draw=white,fill=white,inner sep=0.1pt]{\textcolor{black}{\small 2}} node[pos=0.75,draw=white, fill=white,inner sep=0.1pt] {\textcolor{black}{\small 1}};
\draw(a1) -- (b2) node[pos=0.25,draw=white,fill=white,inner sep=0.1pt]{\textcolor{black}{\small 3}} node[pos=0.7,draw=white, fill=white,inner sep=0.1pt] {\textcolor{black}{\small 3}};
\draw(a2) -- (b2) node[pos=0.25,draw=white,fill=white,inner sep=0.1pt]{\textcolor{black}{\small 1}} node[pos=0.75,draw=white, fill=white,inner sep=0.1pt] {\textcolor{black}{\small 2}};
\draw(a2) -- (b1) node[pos=0.25,draw=white,fill=white,inner sep=0.1pt]{\textcolor{black}{\small 2}} node[pos=0.75,draw=white, fill=white,inner sep=0.1pt] {\textcolor{black}{\small 2}};
\draw(a2) -- (b3) node[pos=0.25,draw=white,fill=white,inner sep=0.1pt]{\textcolor{black}{\small 3}} node[pos=0.75,draw=white, fill=white,inner sep=0.1pt] {\textcolor{black}{\small 2}};
\draw(a3) -- (b2) node[pos=0.25,draw=white,fill=white,inner sep=0.1pt]{\textcolor{black}{\small 1}} node[pos=0.75,draw=white, fill=white,inner sep=0.1pt] {\textcolor{black}{\small 1}};
\draw(a4) -- (b4) node[pos=0.25,draw=white,fill=white,inner sep=0.1pt]{\textcolor{black}{\small 1}} node[pos=0.75,draw=white, fill=white,inner sep=0.1pt] {\textcolor{black}{\small 1}};
\draw(a4) -- (b1) node[pos=0.25,draw=white,fill=white,inner sep=0.1pt]{\textcolor{black}{\small 2}} node[pos=0.85,draw=white, fill=white,inner sep=0.1pt] {\textcolor{black}{\small 3}};
\end{tikzpicture}
	\caption{Counter-example for the generalized approach based on the one used in~\cite{NNRS21}}
	\label{Fig:CounterEx}
\end{figure}

The overview of the reduced instance corresponding to the instance shown in Figure~\ref{Fig:CounterEx} is described as follows. We have $\S=\sum_{a\in\AA}q^-(a)=5$, so each $a\in\AALQ$ have $5+1=6$ level-copies $a^0,a^1,\ldots,a^5$ and $a\in\AA\setminus\AALQ$ have one level-copy $a^0$. Similarly, we have $\T=\sum_{b\in\BB}q^-(b)=2$, so each $b\in\BBLQ$ have $2+2=4$ level-copies $b^0,\ldots,b^3$. Each $b\in\BB\setminus\BBLQ$ have two level-copies $b^0$ and $b^1$. Quota of each level-copy $a^i$ for a vertex $a\in\AA$ is $q^+(a)$ if $i=0$ and $q^-(a)$ otherwise. Quota of each level-copy $b^i$ for a vertex $b\in\BB$ is $q^+(b)$ if $i\le 1$ and $q^-(b)$ otherwise. The number of dummy vertices associated with each level-copy is equal to the quota of that level-copy. Thus, the reduced instance contains level-copies $b_1^0$ and $b_1^1$ with capacities equal to three, and level-copies $b_1^2,b_1^3$ with capacities equal to one. The stable matching of the reduced instance is $M_0=\{(b_1^0,a_1^0), (b_1^0,a_1^1), (b_1^0,a_2^1), (b_2^0,a_2^1), (b_2^0,a_3^0), (b_3^0,a_1^0), (b_4^0,a_4^0) \}$. Note that $q^-(a_1)=q^+(a_1)=3$ and $a_1$ is matched to three different level-copies in the stable matching $M$. But $b_1^0$ is matched to two different level-copies $a_1^0$ and $a_1^1$ of the same vertex $a_1$. Thus, the matching $M_0$ cannot be mapped to a feasible matching in the original instance although at least two feasible matchings $M_1=\{(b_1,a_1), (b_1,a_2), (b_2,a_1), (b_2,a_2), (b_3,a_1), (b_4,a_4)\}$ and $M_2= \{(b_1,a_1), (b_1,a_2),$
$(b_1,a_4), (b_2,a_1), (b_2,a_2), (b_3,a_1)\}$ exist. Algorithm~\ref{algo:maxPopSC2LQ} outputs the matching $M_1$, which is indeed popular amongst feasible/critical matchings.

\input{proposal}

%% file: proposal.tex
\begin{table}
\centering
\begin{tabular}{| p{1.7 cm} | p{0.5 cm} |p{0.9cm} | p{0.5cm} |p{0.8cm} | p{1.8cm} | p{4.6cm}|}
\hline
      \textbf{Proposal number}&\textbf{$a^\ell$}  & \textbf{$c(a^\ell)$} & \textbf{$b$ }& $c(b)$ & \textbf{Rejected vertex} & \textbf{$M$}  \\
      \hline
      \hline \\ [-1em]
      1. &$a_1^0$ & 2 & $b_2$ & $1$& $-$ & $\{(a_1^0,b_2)\}$ \\ [1pt]
      \hline \\ [-1em]
      2. &$a_2^0$ & 2 & $b_2$ & $1$& $a_2^0$ & $\{(a_1^0,b_2)\}$\\ [1pt]
      \hline \\ [-1em]
      3. &$a_3^0$ & 1 & $b_2$ & $1$& $a_1^0$ & $\{(a_3^0,b_2)\}$ \\ [1pt]
      \hline \\ [-1em]
      4. &$a_1^1$ & 2 & $b_1$ & $1$& $-$ & $\{(a_1^1,b_1),(a_3^0,b_2)\}$ \\ [1pt]
      \hline \\ [-1em]
      5. &$a_2^1$ & 2 & $b_1$ & $1$& $a_2^1$ & $\{(a_1^1,b_1),(a_3^0,b_2)\}$\\ [1pt]
      \hline \\ [-1em]
      6. & $a_1^1$ & 2 & $b_2$ & $1$& $a_3^0$ & $\{(a_1^1,b_1),(a_1^1,b_2)\}$ \\ [1pt]
      \hline \\ [-1em]
      7. & $a_2^1$ & 2 & $b_2$ & $2$& $-$ & $\{(a_1^1,b_1),(a_1^1,b_2),(a_2^1,b_2)\}$\\ [1pt]
      \hline \\ [-1em]
      8. & $a_3^1$ & 1 & $b_2$ & $2$& $a_2^1$ & $\{(a_1^1,b_1),(a_1^1,b_2),(a_3^1,b_2)\}$ \\ [1pt]
      \hline \\ [-1em]
      9. & $a_2^2$ & 2 & $b_1$ & $1$& $a_1^1$ & $\{(a_2^2,b_1),(a_1^1,b_2),(a_3^1,b_2)\}$ \\ [1pt]
      \hline \\ [-1em]
      10. & $a_1^2$ & 2 & $b_1$ & $1$& $a_2^2$ & $\{(a_1^2,b_1),(a_1^1,b_2),(a_3^1,b_2)\}$ \\ [1pt]
      \hline \\ [-1em]
      11. & $a_2^2$ & 2 & $b_2$ & $2$& $a_1^1$ & $\{(a_1^2,b_1),(a_2^2,b_2),(a_3^1,b_2)\}$ \\ [1pt]
      \hline \\ [-1em]
      12. &$a_1^2$ & 2 & $b_2$ & $2$& $a_3^1$ & $\{(a_1^2,b_1),(a_2^2,b_2),(a_1^2,b_2)\}$ \\ [1pt]
      \hline \\ [-1em]
      13.  & $a_3^2$ & 1 & $b_2$ & $2$& $a_2^2$ & $\{(a_1^2,b_1),(a_3^2,b_2),(a_1^2,b_2)\}$ \\ [1pt]
      \hline \\ [-1em]
      14. & $a_2^3$ & 2 & $b_1$ & $1$& $a_1^2$ & $\{(a_2^3,b_1),(a_3^2,b_2),(a_1^2,b_2)\}$ \\ [1pt]
      \hline \\ [-1em]
      15. & $a_2^3$ & 2 & $b_2$ & $2$& $a_1^2$ & $\{(a_2^3,b_1),(a_3^2,b_2),(a_2^3,b_2)\}$ \\ [1pt]
      \hline \\ [-1em]
      16. & $a_1^3$ & 1 & $b_1$ & $1$& $a_2^3$ & $\{(a_1^3,b_1),(a_3^2,b_2),(a_2^3,b_2)\}$ \\ [1pt]
      \hline \\ [-1em]
      17. & $a_2^4$ & 2 & $b_1$ & $1$& $a_1^3$ & $\{(a_2^4,b_1),(a_3^2,b_2),(a_2^3,b_2)\}$ \\ [1pt]
      \hline \\ [-1em]
      18. & $a_1^3$ & 1 & $b_2$ & $2$& $a_3^2$ & $\{(a_2^4,b_1),(a_1^3,b_2),(a_2^3,b_2)\}$ \\ [1pt]
      \hline \\ [-1em]
      19. & $a_3^3$ & 1 & $b_2$ & $2$& $a_2^3$ & $\{(a_2^4,b_1),(a_1^3,b_2),(a_3^3,b_2)\}$ \\ [1pt]
      \hline \\ [-1em]
      20. & $a_2^4$ & 2 & $b_2$ & $2$& $a_1^3$ & $\{(a_2^4,b_1),(a_2^4,b_2),(a_3^3,b_2)\}$ \\ [1pt]
      \hline \\ [-1em]
      21. & $a_1^4$ & 1 & $b_1$ & $1$& $a_2^4$ & $\{(a_1^4,b_1),(a_2^4,b_2),(a_3^3,b_2)\}$ \\ [1pt]
      \hline \\ [-1em]
      22. & $a_2^5$ & 2 & $b_1$ & $1$& $a_1^4$ & $\{(a_2^5,b_1),(a_2^4,b_2),(a_3^3,b_2)\}$ \\ [1pt]
      \hline \\ [-1em]
      23. &$a_1^4$ & 1 & $b_2$ & $2$& $a_3^3$ & $\{(a_2^5,b_1),(a_2^4,b_2),(a_1^4,b_2)\}$ \\ [1pt]
      \hline \\ [-1em]
      24.  & $a_3^4$ & 1 & $b_2$ & $2$& $a_2^4$ & $\{(a_2^5,b_1),(a_3^4,b_2),(a_1^4,b_2)\}$ \\ [1pt]
      \hline \\ [-1em]
      25. & $a_2^5$ & 2 & $b_2$ & $2$& $a_1^4$ & $\{(a_2^5,b_1),(a_3^4,b_2),(a_2^5,b_2)\}$ \\ [1pt]
      \hline \\ [-1em]
      26. & $a_1^5$ & 1 & $b_1$ & $1$& $a_2^5$ & $\{(a_1^5,b_1),(a_3^4,b_2),(a_2^5,b_2)\}$ \\ [1pt]
      \hline \\ [-1em]
      27. & $a_2^6$ & 2 & $b_1$ & $1$& $a_1^5$ & $\{(a_2^6,b_1),(a_3^4,b_2),(a_2^5,b_2)\}$ \\ [1pt]
      \hline \\ [-1em]
      28. & $a_1^5$ & 1 & $b_2$ & $2$& $a_3^4$ & $\{(a_1^5,b_2),(a_2^6,b_1),(a_2^5,b_2)\}$ \\ [1pt]
      \hline \\ [-1em]
      29. & $a_3^5$ & 1 & $b_2$ & $2$& $a_2^5$ & $\{(a_1^5,b_2),(a_2^6,b_1),(a_3^5,b_2)\}$ \\ [1pt]
      \hline \\ [-1em]
      30. & $a_2^6$ & 2 & $b_2$ & $2$& $a_1^5$ & $\{(a_2^6,b_1),(a_2^6,b_2),(a_3^5,b_2)\}$ \\ [1pt]
      \hline \\ [-1em]
      31. & $a_1^6$ & 1 & $b_1$ & $1$& $a_2^6$ & $\{(a_1^6,b_1),(a_2^6,b_2),(a_3^5,b_2)\}$ \\ [1pt]
\hline     
\hline
\end{tabular}
\caption{A possible proposal sequence of Algorithm~\ref{algo:maxPopSC2LQ} for the instance shown in Figure~\ref{fig:exCritical}.}
\label{tab:propSeq}
\end{table}